\newcommand{\degmax}{\operatorname{deg}_{\max}}
\newcommand{\ELDS}{\mathcal{E}_{\mathrm{LDS}}}
\newcommand{\PLDS}{\mathcal{P}_{\mathrm{LDS}}}

\newcommand{\ELCDS}{\mathcal{E}_{\mathrm{LCDS}}}
\newcommand{\PLCDS}{\mathcal{P}_{\mathrm{LCDS}}}

\newcommand{\WLDS}{\mathcal{W}_{\mathrm{LDS}}}
\newcommand{\DLDS}{\mathcal{D}_{\mathrm{LDS}}}

\newcommand{\WLCDS}{\mathcal{W}_{\mathrm{LCDS}}}
\newcommand{\DLCDS}{\mathcal{D}_{\mathrm{LCDS}}}

\documentclass[twoside,leqno,twocolumn]{article}

\usepackage{algorithm}
\usepackage{algpseudocode}
\usepackage{caption}
\usepackage{subcaption}
\usepackage{float}

\algrenewcommand\algorithmicrequire{\textbf{Input:}}
\algrenewcommand\algorithmicensure{\textbf{Output:}}

\usepackage[letterpaper]{geometry}

\usepackage{siamproceedings}

\usepackage[T1]{fontenc}
\usepackage{amsfonts}
\usepackage{graphicx}
\usepackage{epstopdf}
\usepackage{enumitem}
\ifpdf
  \DeclareGraphicsExtensions{.eps,.pdf,.png,.jpg}
\else
  \DeclareGraphicsExtensions{.eps}
\fi

\newsiamremark{remark}{Remark}
\newsiamremark{hypothesis}{Hypothesis}
\crefname{hypothesis}{Hypothesis}{Hypotheses}
\newsiamthm{claim}{Claim}

\usepackage[textwidth=2cm,textsize=footnotesize]{todonotes}

\usepackage{comment}

\usepackage{amsopn}
\usepackage{mathtools}

\algnewcommand\algorithmicparfor{\textbf{par for}}
\algblockdefx{ParFor}{EndParFor}%
  [1]{\algorithmicparfor\ #1\ \algorithmicdo}%
  {\algorithmicend\ \algorithmicfor}%

\algnewcommand\algorithmicpforeach{\textbf{par for each}}
\algblockdefx{ParForEach}{EndParForEach}%
  [1]{\algorithmicpforeach\ #1\ \algorithmicdo}%
  {\algorithmicend\ \algorithmicfor}%

\DeclarePairedDelimiter\ceil{\lceil}{\rceil}

\newlength{\datasettablewidth}
\usepackage{booktabs,tabularx}
\setlist{topsep=2pt,itemsep=1pt,partopsep=2pt,parsep=2pt}
\begin{document}
\setlength{\textfloatsep}{4pt}

\newcommand\relatedversion{}
\renewcommand\relatedversion{\thanks{All code is publicly available at: \protect\url{https://doi.org/10.5281/zenodo.23433263}}} 
\onecolumn
\title{\Large An experimental evaluation of static $k$-center clustering algorithms on graphs%
\thanks{This project has received funding from the European Research Council (ERC) under the European Union's Horizon 2020 research and innovation programme (grant agreement No 947702). All code is publicly available at: \protect\url{https://doi.org/10.5281/zenodo.23433263}}}

\author{Sebastian Forster\thanks{University of Salzburg, Austria (\email{forster@cs.sbg.ac.at})}
    \and Martin Grösbacher\thanks{University of Salzburg, Austria (\email{martin.groesbacher@plus.ac.at})}}
\date{}

\maketitle

\fancyfoot[R]{\scriptsize{Copyright \textcopyright\ 2027\\
Copyright for this paper is retained by authors}}

\begin{center}
\begin{minipage}{0.75\textwidth}
    \begin{abstract} The $k$-center problem is one of the most fundamental clustering problems and has applications in a wide range of areas, including facility location, network design, data mining, and machine learning. Given a set of objects together with a distance measure, the goal is to select $k$ centers such that the maximum distance from any object to its closest center is minimized. The $k$-center problem has been studied extensively from a theoretical perspective, both in metric spaces and on graphs. While there have been experimental studies for $k$-center clustering in metric spaces, a comprehensive experimental comparison of graph-based $k$-center algorithms on large-scale graphs is still lacking. In this paper, we bridge this gap by evaluating a set of algorithms ranging from classical approaches such as Gonzalez's greedy algorithm [TCS 1985], an algorithm based on the framework established by Hochbaum and Shmoys [Math. Oper. Res. 1985], to more recent methods, including algorithms by Thorup [SICOMP 2001] and Abboud et al. [SOSA 2023]. Our study considers both running time, solution quality and the effects of parallelism on a variety of synthetic and real-world datasets. Our goal is to provide a comprehensive comparison of state-of-the-art $k$-center algorithms on graphs and identify the strengths and weaknesses of different approaches in practice.
\end{abstract}
\end{minipage}
\end{center}
\twocolumn
\section{Introduction}
\label{sec:intro}


Analyzing the structure of large-scale networks is a fundamental problem in many fields. One important task is to partition the underlying data into meaningful subgroups. This process is naturally captured by clustering problems. These problems have been extensively studied in data mining and machine learning~\cite{Jain, Reddy, Kogan}. The goal of clustering problems is to group input data together in structures called clusters such that similar data ends up in the same cluster. Among the many clustering objectives, $k$-clustering methods represent each cluster by one of $k$ centers and each data point is assigned to the cluster of its closest center. The $k$-center problem is concerned with selecting its centers such that the maximum distance of every data point to its closest center is minimized. Formally, this is an optimization problem of the following form: Given a metric space $M=(P,\mathrm{dist})$ and an integer $k \geq 1$, the objective is to compute a set of centers $C \subseteq P$ with $|C|=k$ minimizing $f_M(C)=\max_{p\in P}\mathrm{dist}(p,C)$, where $\mathrm{dist}(p,C)=\min_{q\in C}\mathrm{dist}(p,q)$.

This problem is $\mathsf{NP}$-hard~\cite{HsuN79} and thus TCS research generally focuses on finding efficient approximation algorithms~\cite{Gonzalez, HS, Plesnik87, Dyer}. The $k$-center problem is commonly defined on arbitrary metric spaces. In this general setting, often called `metric $k$-center' or `vertex $k$-center', it is usually assumed that the distances are given to the algorithm explicitly as part of the input or that distances can be queried upon request. However, in transportation, communication, social and web network analysis, data is often inherently represented as a graph, where distances correspond to shortest-path lengths rather than explicitly given pairwise distances. Consequently, the graph $k$-center problem has been studied alongside its metric counterpart from theoretical and algorithmic perspectives~\cite{Thorup, Abboud, ForsterNRS26}. Although graph metrics satisfy the triangle inequality and therefore constitute a special case of metric spaces, they introduce a fundamental algorithmic difference. As a result, the practical performance of graph $k$-center algorithms depends not only on the quality of the approximation algorithm itself, but also on how efficiently shortest-path computations are performed. Despite recurring progress of theoretical research on graph $k$-center there is comparatively little empirical work evaluating these algorithms on modern large-scale graph datasets. This paper aims to fill this gap by providing a comprehensive experimental study of state-of-the-art graph $k$-center algorithms on graphs.

\subsection{Related Work}


Note that for the sake of clearer exposition, we report running times under the assumption that the weights are integers in the range $ 1 $ to $ n^c $ for some constant $ c $.
For super-polynomial weights, general techniques for reducing the weight range at the expense of a mild increase in approximation guarantee exist (see~\cite{Abboud} for a brief discussion).
We furthermore employ $ \tilde O (\cdot) $-notation suppresses factors polylogarithmic in $ n $ (the number of nodes).

Classic approximation algorithms for $k$-center have originally considered the ``metric'' version of the problem. 
Both the farthest-point heuristic (discovered independently by Gonzalez~\cite{Gonzalez} and Dyer and Frieze~\cite{Dyer}) and the covering heuristic (discovered independently by Hochbaum and Shmoys~\cite{HS} and Plesník~\cite{Plesnik87}) compute a solution that exceeds the optimal value by a factor of at most $ 2 $ in polynomial time.
This approximation factor is tight unless $ \mathsf{P} = \mathsf{NP} $~\cite{HsuN79}.\footnote{Note that such hardness results directly carry over to the graph setting because every metric space with $ n $ points can be viewed as a complete graph with $ n^2 $ edges.}
In cases where constant-time query access to the distances of the metric space is supported, these two classic algorithms have a running time of $ O (k n) $ and $ \tilde O (k n) $, respectively.\footnote{Otherwise, a quadratic-time preprocessing step is required for reading the whole input.}
Note that any algorithm formulated for the metric setting can readily be applied to graphs after performing an (expensive, $ \Omega (n^2) $-time) all-pairs shortest path computation in a preprocessing step.\footnote{As pointed out by Thorup~\cite{Thorup}, distance approximation algorithms can be used to gain efficiency in such a preprocessing step at the cost of overall accuracy.}
The two classic algorithms can however be directly implemented more efficiently in $ \tilde O (km + n) $ time by employing Dijkstra's algorithm (or any other comparably fast single-source shortest path (SSSP) algorithm) for a graph with $ n $ nodes and $ m $ edges.
It can further be shown that at least $ \Omega (k n) $ distance queries are necessary in the metric setting~\cite{BateniEFHJMW23}, but this lower bound does not apply to the graph setting.

Indeed, Thorup~\cite{Thorup} designed a $2$-approximation algorithm with near-linear running time, i.e., time $ \tilde O (m + n) $.
Thorup in particular exploits the following near-linear time primitive that cannot be implemented as efficiently in other metric spaces: Given a set of source nodes $ S \subseteq V $, compute the distance of every node to its closest node in $ S $ by a straightforward modification of Dijkstra's algorithm.
Another interesting aspect about Thorup's algorithm is that it eliminates the inherent ``chain of sequentiality'' present in the classic algorithms that arises from the iterative choice of the $ k $ centers.

Recently, the graph setting has received further attention.
First, Abboud, Cohen-Addad, Lee, and Manurangsi~\cite{Abboud} showed that the two classic algorithmic approaches mentioned above can be implemented in near-linear time in the graph setting due to an elegant randomization trick; for the Hochbaum-Shmoys algorithm this approach worsens the approximation guarantee by a $ (1 + \epsilon) $-factor for a tuneable parameter~$ \epsilon $ with (inverse) linear influence on the running time.
Second, Forster, Nazari, Rao, and Skarlatos~\cite{ForsterNRS26} observed that the two classic algorithmic approaches can be implemented efficiently by employing a partially dynamic single-source shortest paths data structure.
In particular, the partially dynamic algorithm of Górkiewicz and Karczmarz~\cite{GorkiewiczK25} implies a nearly-linear time implementation  at the cost of worsening the approximation guarantee by a $ (1 + \epsilon) $-factor for a tuneable parameter~$ \epsilon $ with (inverse) linear influence on the running time.
Note that these recent works again exhibit the aforementioned ``chain of sequentiality'' in choosing the centers.

Apart from the study of Kralev et al.~\cite{KralevKAC21}, which compares an exact and an approximation algorithm on randomly generated graphs with at most $150$ nodes, we are not aware of any experimental studies targeting the $k$-center problem in the graph setting.
Existing studies~\cite{MihelicR05,Garcia-DiazMMHS19} for the related metric $k$-center problem (in which the all-pairs distances are given explicitly as part of the input) focus approximation guarantee and have been conducted for moderate input sizes.

\subsection{Research Questions}

As explained above, the $k$-center problem on graphs has become a separate object of study with several potential use cases. Yet, to the best of our knowledge, a systematic study of running time performance is lacking. In this paper, we address this research gap by engineering scalable implementations of existing algorithms and introducing an implementation-oriented simplification of Thorup's algorithm. In doing so, we in particular consider parallelization issues to support large input graphs on modern hardware. All algorithms in this paper were therefore implemented both with sequential and parallel support via parallel SSSP routines.

The main goal of our work is to answer the following questions:

\begin{enumerate}
\item Which algorithms achieve the best running time on different graph classes?
\item How do the solution qualities of the algorithms compare?
\item How well do the algorithms scale with parallelism?
\end{enumerate}

To answer these questions, we perform an extensive experimental evaluation on synthetic random graphs as well as real-world instances.
To keep this study feasible, we (a) limit the scope of our investigations to polynomial-time algorithms with a (close to) optimal multiplicative approximation factor of $2$ and (b) rely on established implementations of single-source shortest path algorithms. Algorithms included in this paper are well known approaches such as the algorithm of Gonzalez \cite{Gonzalez}, a recent improvement over the classic algorithm of Hochbaum \& Shmoys \cite{HS} by Abboud et al.~\cite{Abboud}, an algorithm by Thorup~\cite{Thorup} as well as an $\alpha$-approximate version of Gonzalez' algorithm by Abboud et al.~\cite{Abboud}. All algorithms were implemented in C++ using the graph-based benchmark suite (GBBS) by Dhulipala et al. \cite{GBBS}.

\subsection{Outline}

The rest of this paper is organized as follows:
In section~\ref{sec:preliminaries} we provide preliminaries and discuss our parallel model. Sections~\ref{sec:mis} and \ref{sec:farthest} explain the algorithms and their implementations. Sections~\ref{sec:setup} and~\ref{sec:results} treat our experimental setup and results and section~\ref{sec:conclusion} contains the conclusion.
\section{Preliminaries}
\label{sec:preliminaries}
\subsection{Basic Concepts}
Let $G = (V, E, w)$ be an undirected (weighted) graph with vertex set $V$, edge set $E$, and weight function $w : E \rightarrow \mathbb{N}_{>0}$ assigning a positive integer weight to each edge. Throughout this paper, let $n = |V|$ denote the number of vertices, $m = |E|$ the number of edges and let $N$ denote the highest weight in the graph which we assume to be polynomial in $n$. This assumption is justified by the observation that $N << n$ in massive real-world data sets with a large number of nodes. Furthermore, we define the density of $G$ as $\rho = \frac{m}{n}$. We write $\deg(v)$ to refer to the degree of the vertex $v$, i.e. the number of its incident edges and $\degmax$ to refer to the highest degree of any $v \in V$ of $G$. Let $d(u,v)$ denote the weight of the shortest path between $u$ and $v$ in $G$, sometimes also referred to as the distance from $u$ to $v$. Since the graphs considered in this paper are undirected, $d(u,v)$ is symmetric. The $k$-center problem on graphs is defined as:
\begin{definition}[$k$-Center Problem on Graphs]
    Given an undirected (weighted) graph $G=(V,E,w)$ and an integer $k$, find a set of centers $C \subseteq V$ with $|C| \leq k$ minimizing
    \[
    \max_{v \in V} d(v,C),
    \quad\text{where}\quad
    d(v,C)=\min_{c\in C} d(v,c).
    \]
\end{definition}

The solution to the $k$-center problem returned is often referred to as the radius of the clusters produced. For a vertex $u \in V$, let $d_u$ denote its current distance to the nearest center. During the execution of an algorithm, $d_u$ refers to the value maintained by the algorithm and is updated whenever a closer center is found. We say that an event occurs with high probability (w.h.p.), if it occurs with probability at least $1 - \frac{1}{n^c}$ for any $c \geq 1$.

\subsection{Parallel Model and Primitives}

In this paper, we analyze the running time of parallel algorithms in the work-depth model~\cite{Blelloch96,SandersS24}, which is commonly used to analyze fork-join style parallel algorithms. Work is the total number of operations performed (equivalently, the sequential running time). Depth (span) is the length of the longest chain of sequential dependencies. By Brent's theorem \cite{Brent}, the running time of a parallel algorithm in the work-depth model on $P$ processors satisfies $T_P = O(\frac{W}{P} + D)$ where $W$ denotes the total work, $P$ denotes the number of available processors, $D$ denotes the depth and $T_P$ denotes the running time of the algorithm using $P$ processors. This model is a shared-memory extension of the sequential RAM model and has been commonly used to develop and analyze parallel algorithms. A work-efficient parallel algorithm is one where the work asymptotically matches the best-known sequential time complexity for the problem.

\section{SSSP algorithms}
\subsection{Parallel and Single-Core Routines}
\label{sec:sssp}

In this section we describe the different shortest path computations used by our $k$-center algorithms for distance computation. Many graph-based $k$-center algorithms repeatedly solve the Single-Source Shortest Path (SSSP) problem as a subroutine. Consequently, the overall performance of these algorithms is heavily influenced by the efficiency of the underlying shortest-path computations. Since SSSP is a well-studied problem in graph algorithms, efficient sequential and parallel implementations are readily available.

Depending on whether the graph is weighted or unweighted and whether a sequential or parallel execution is used, we employ four different SSSP routines:

\begin{enumerate}
    \item Parallel, unweighted: \textbf{Parallel Breadth-First Search}
    \item Parallel, weighted: \textbf{$\Delta$-stepping}
    \item Single-core, unweighted: \textbf{Breadth-First Search}
    \item Single-core, weighted: \textbf{Dijkstra}
\end{enumerate}

All four routines are extended by two modifications that are used throughout this paper. First, we employ \emph{lazy updates}. Since all algorithms repeatedly perform SSSP computations from different centers, they maintain a global distance array which for every vertex $v$ stores its current distance $d_v$ to the closest center selected so far. During an SSSP computation, a vertex is only inserted into the SSSP data structure if its distance $d_v$ is improved. Second, we support an optional distance cap. Given a source vertex $u$ and a radius $r$, an edge $(u,v)$ is only relaxed if $d_u + w(u,v) \leq r$.

Breadth-First Search (both sequential and parallel) and Dijkstra's algorithm are standard algorithms whose details are well-known.
For the parallel execution on weighted graphs, we use $\Delta$-stepping~\cite{Meyer}.
We are aware of the existence of other algorithms that use the stepping framework~\cite{Dong21}.
While they show improvements for certain graph types, $\Delta$-stepping still remains competitive and is widely used. Since our experimental dataset contains multiple graph classes, we make vanilla $\Delta$-stepping our algorithm of choice. In the unweighted setting, the effects of parallelization are relatively intuitive and we therefore do not evaluate these explicitly in our study. A review of $\Delta$-stepping and its subsequent notation is contained in~\cref{sec:ds} of the appendix.

More generally, other approaches for sequential and parallel single-source shortest path computation~\cite{Thorup99,AndoniSZ20,Li20,RozhonGHZL22,Rozhon23,Fineman23,DuanMSY23,DuanMMSY25,GrunauKZ25,CastroCF25,DuanMSY26,AshvinkumarBGS26} (a) precompute global structures which adds a certain baseline overhead and seems to not interact well with locality aspects of $k$-center and/or (b) have no publicly available scalable implementations.
\section{Maximal Distance-$r$ Independent Set Algorithms}
\label{sec:mis}
\subsection{Abboud MIS}
\label{sec:abboud}

The first $k$-center clustering algorithm we implemented is based on the well-known $2$-approximation algorithm by Hochbaum and Shmoys~\cite{HS}. It takes a graph $G = (V, E, w)$ along with a parameter $k$ as an input and outputs a $2$-approximation to the $k$-center solution by computing the smallest radius $r$ that is sufficient to cover $G$ using at most $k$ centers. A maximal distance-$r$ independent set algorithm is used as a feasibility check for a given $r$ and depending on the output, the radius is increased or decreased according to a binary search. Abboud et al. showed that a maximal distance-$r$ independent set algorithm can be implemented in $O(n \log^2 n + m \log n)$ time in expectation by combining random center selection with lazy Dijkstra updates.\footnote{It should be mentioned, that Abboud et al.\ employ a rounding technique such that each edge weight is an integer between $1$ and $O(\frac{n}{\epsilon})^{O(1)}$ which guarantees polynomially bounded edge weights but causes their solution to lose an approximation factor of $1 + \epsilon$. Since, as mentioned in the preliminaries, we observe that the highest weight in real world datasets is often much smaller than $n$, we do not employ this rounding scheme for any of our algorithms.} Note that their running time is independent of $k$. They key observation is that each vertex and consequently its incident edges are scanned only $O(\log n)$ times in expectation. We now show that the same charging argument remains valid whether we use our Lazy Capped Dijkstra or the Lazy Capped $\Delta$-stepping routine, because it depends only on the sequence of successful distance decreases and not on how the SSSP routine propagates those. Thus, we will now analyze the feasibility check employed by Abboud et al.'s algorithm, which henceforth will simply be called Abboud MIS.

\begin{algorithm}
\caption{\textsc{AbboudFeasibilityCheck}$(G, r)$}\label{alg:mis}
\begin{algorithmic}[1] 

\Require Graph $G=(V,E)$, radius $r$
\Ensure Maximal distance-$r$ independent set $S$

\State $S \gets \emptyset$ \Comment{output set}
\State $T \gets V$ \Comment{remaining center candidates}
\State $d[v] \gets \infty \;\; \forall v \in V$ \Comment{global distance array}

\While{$T \neq \emptyset$}
    \State $v \gets$ random vertex from $T$ \Comment{pick center}
    \State $T \gets T \setminus \{v\}$
    \If{$d[v] > r$} 
        \State $S \gets S \cup \{v\}$ \Comment{add vertex to solution}
        \State \textproc{SSSP Routine}$(G, v, r, d)$ \Comment{update distances}
    \EndIf
\EndWhile

\State \textbf{return} $S$
\end{algorithmic}
\end{algorithm}

As mentioned before, Algorithm~\ref{alg:mis} serves as the feasibility check for the maximal distance-$r$ independent set algorithm by Abboud et al.: Given a radius $r$, it will return the set of centers and therefore also the number of centers needed to cover the whole graph with the given $r$. In the weighted parallel case, it calls Lazy Capped $\Delta$-stepping as a subroutine once per selected center. The following lemma shows that each vertex experiences only $O(\log n)$ successful distance decreases in expectation. Consequently, the exploration work accumulated over all SSSP executions is only an $O(\log n)$ factor larger than the exploration work of a single Lazy Capped $\Delta$-stepping execution, whereas the phase overhead accumulates once per selected center since the SSSP computations are invoked sequentially. Let $S_r$ denote the output set produced by \Cref{alg:mis} on input $(G,r)$ and let $k_r := |S_r|$ denote the number of centers selected by Algorithm $3$ for this input radius. 

The following lemma bounds the expected work and depth of Algorithm~\ref{alg:mis} when Lazy Capped $\Delta$-stepping is used as the underlying SSSP routine.

\begin{lemma}
\label{lem:abboud-feasibility}

Let $k_r$ denote the number of centers selected by
Algorithm~\ref{alg:mis}. When Lazy Capped $\Delta$-stepping is used as the
underlying SSSP routine, the expected work is

\[
O(\ELCDS\log n+k_r\PLCDS).
\]

and the depth is

\[
O(k_r\DLCDS).
\]

\end{lemma}

The proof and the implementation details are contained in~\cref{sec:abboudProof} of the appendix.
\subsection{Simplifying Thorup's Algorithm}
\label{sec:thorup}

Thorup's algorithm follows the Hochbaum-Shmoys framework which by binary search finds the smallest $r$ such that a distance-$r$ independent set of size $\leq k$ can be found \cite{Thorup}.
In each iteration, it starts with $U := \emptyset$, the initially empty distance-$r$ independent set, and the set of uncovered nodes $W := V$.
The algorithm then iteratively grows the independent set by repeatedly performing the following steps until $ W $ is empty:
\begin{itemize}
    \item Compute $\delta$, a $(1 \pm \epsilon)$-approximation of $ \delta^* $, the maximum number of elements of $W$ contained in any distance-$r$ ball.
    \item Pick a random subset $R$ of $W$ of size $\tfrac{|W|}{\delta}$.
    \item Compute the subset $T \subseteq R$ of ``isolated'' nodes having no other nodes of $R$ at distance $\leq r$.
    \item Add $T$ to $U$ and remove all nodes from $W$ at distance $\leq r$ to any node in $U$.
\end{itemize}

As a result, Thorup's algorithm produces a $2$-approximation to the $k$-center problem. It is clear that the set $U$ constructed in this way is indeed a maximal distance-$r$ independent set.
The size-estimation step can be done with Cohen's near-linear time size-estimation algorithm~\cite{Cohen97} and Thorup provides implementation details to argue how the rest of each iteration can be performed in near-linear time.

To show that the number of iterations is polylogarithmic in $n$ in expectation, Thorup argues that it takes $O(\log n)$ iterations in expectation to reduce $\delta^*$, the maximum number of elements of $W$ contained in any distance-$r$ ball, by a constant factor. Since $\delta^*$ can undergo only $O(\log n)$ such reductions, this gives an expected bound of $O(\log^2 n)$ on the total number of iterations. While Thorup's algorithm achieves a near-linear running time, one of its key ingredients is Cohen's approximate neighborhood-size estimation data structure. Although elegant from a theoretical perspective, this routine is comparatively involved, introduces polylogarithmic running-time overheads, and, to the best of our knowledge, no publicly available implementation exists. Our goal is therefore to replace this component by a much simpler implementation-friendly alternative while preserving the same asymptotic running time guarantees. The resulting algorithm is considerably easier to implement, requires only repeated shortest-path computations, and has no worse asymptotic running time than Thorup's original algorithm.

\begin{algorithm}
\caption{\textsc{Simpl.ThorupFeasibilityCheck}$(G,r)$}
\label{alg:thorup}
\begin{algorithmic}[1]

\Require Graph $G=(V,E)$, radius $r$
\Ensure Maximal distance-$r$ independent set $U$

\State $U \gets \emptyset$
\State $W \gets V$
\State $\delta \gets (1 + \epsilon)^{\lfloor \log_{1 + \epsilon} n \rfloor + 1}$

\While{$W \neq \emptyset$ \textbf{and} $\delta \geq 1$}
\label{alg:thorupWhile}
    \For{$i \gets 1$ \textbf{to} $c\log n$}

        \State Select a random subset $R \subseteq W$ of size $\max\{1,\lceil |W|/\delta\rceil\}$

        \State $T \gets \{v\in R \mid \Gamma_{\leq r}(v)\cap R=\{v\}\}$

        \State $U \gets U \cup T$

        \State Remove from $W$ all vertices within distance $r$ of $U$

        \If{$W = \emptyset$}
            \State \textbf{return} $U$
        \EndIf

    \EndFor

    \State $\delta \gets \frac{\delta}{1+\epsilon}$

\EndWhile
\label{alg:thorupEndWhile}
\State \textbf{return}
\textproc{Fallback MIS}$(G,r,U,W,d)$

\end{algorithmic}
\end{algorithm}
\vspace{3pt}
Our algorithm replaces Cohen's repeated size-estimation routine by a geometrically decreasing guess of the sampling parameter $\delta$. Initially, we set $\delta \approx n$ and repeatedly decrease it by a factor of $\frac{1}{1+\epsilon}$ after each ``phase'' (by which we denote one iteration of the outer while-loop in lines~\ref{alg:thorupWhile}--\ref{alg:thorupEndWhile} of Algorithm~\ref{alg:thorup}).
For each phase with a fixed value of $ \delta $, we repeat the sampling of nodes for the independent set $ O (\log n) $ times (whereas Thorup samples only once for a specific estimate $ \delta $); we refer to one iteration of the inner for-loop in lines $5$--$13$ in which this happens as a ``round''.

A closely related guessing argument appears in the distributed beeping-model MIS algorithm of Afek et al.~\cite{Afek}. Although their setting differs from ours, they also avoid explicitly determining the relevant neighborhood size by considering geometrically spaced estimates and analyze the iteration in which the current estimate is within a constant factor of the true scale. We adapt this idea directly for the $k$-center problem using the framework established by Thorup.

In our simplification of Thorup's algorithm, most guesses initially overestimate the unknown value $\delta^*$ and therefore need not contribute to the progress of the algorithm. However, these rounds are harmless. Eventually, the guessed value enters the constant-factor estimation regime considered in Thorup's analysis, where he argues that each heavy vertex, i.e., each vertex $v$ with $|\Gamma_{\leq r}(v) \cap W| \geq 0.6\delta^*$ where $\Gamma_{\leq r}(v)$ is the $r$-neighborhood of $v$, is eliminated from $W$ with constant probability. 
Thus, in expectation a constant fraction of heavy nodes is eliminated in each round of such a phase and after $ O (\log n) $ rounds (the length of each phase), each heavy vertex is eliminated with high probability, which implies that $\delta^*$ is reduced by a constant factor in this phase.
We can summarize this guarantee as follows.

\begin{lemma}\label{lm:Thorup main lemma}
    For some constants $ c $ and $ \epsilon $ and any fixed $ \delta \geq 1 $, the relationship $ \delta^* < \delta \leq (1 + \epsilon) \delta^* $ can hold for at most $ c \log n $ consecutive rounds of Algorithm~\ref{alg:thorup} with high probability.
\end{lemma}

After such a ``successful'' phase, $\delta$ is subsequently decreased geometrically as well and the argument repeats for the new values of $ \delta $ and $ \delta^* $.
The geometric decrease bounds the number of phases by $ O (\log n) $ and by the following invariant we can argue that a successful phase must happen eventually.


\begin{lemma}
\label{lm:thorup}
At the beginning of each phase of Algorithm~\ref{alg:thorup} and immediately after each update of $\delta$, we have $\delta > \delta^*$ with high probability.
\end{lemma}

\begin{theorem}
\label{thm:thorup}
Algorithm~\ref{alg:thorup} terminates after $O(\log^2 n)$ rounds with high probability.
\end{theorem}

The proofs of Lemmas~\ref{lm:Thorup main lemma} and~\ref{lm:thorup} and Theorem~\ref{thm:thorup} as well as the implementation details are contained in~\cref{sec:thorupProofs} of the appendix. To summarize, our simplified algorithm replaces the size-estimation routine by a geometrically decreasing guess of $\delta$. By \Cref{thm:thorup}, one feasibility check performs $O(\log^2 n)$ rounds w.h.p. Each round uses $O(\log n)$ invocations of the underlying Lazy Capped SSSP routine which arise from Thorup's step A.2.3 to identify centers and then remove covered vertices within distance $r$. Thus, one feasibility check takes time $O(\log^3 n \cdot (m + n\log n))$ w.h.p. in the sequential setting using Lazy Capped Dijkstra. Using Lazy Capped $\Delta$-stepping in the parallel setting, one feasibility check has work $W = O\!\left(\log^3 n\, (\ELCDS+\PLCDS)\right)$ and depth $D = O\!\left(\log^3 n\, \DLCDS\right)$ w.h.p. In our experimental evaluation, we use a fixed number of four rounds per phase with $\epsilon = 1$, as motivated by the parameter study in~\cref{par:thorup}. While this configuration improves practical performance, it does not satisfy the \(O(\log n)\) rounds-per-phase assumption required for our high-probability termination guarantee. Correctness is nevertheless preserved by the fallback algorithm.

\section{Farthest-Point Selection Algorithms}
\label{sec:farthest}
\subsection{Gonzalez}
\label{sec:gonzalez}

The algorithm of Gonzalez is arguably one of the most well known approaches to the $k$-center problem \cite{Gonzalez}. A similar approach was developed independently by Dyer and Frieze around the same time \cite{Dyer}. The algorithm is based on farthest-point selection: Given a graph $G = (V, E, w)$ as an input, it outputs a $2$-approximation to the $k$-center solution by picking an arbitrary first center $u$. Then it continues by computing an SSSP solution from $u$ which updates the distances in the graph and then selects $v = \arg\max_{x \in V} d_x$ as the next center. This is repeated until $k$ centers have been picked.

\begin{algorithm}

\caption{\textsc{Gonzalez}$(G, k)$}
\label{alg:gonzalez}
\begin{algorithmic}[1]

\Require Graph $G=(V,E)$, number of centers $k$
\Ensure Center set $C$

\State Pick a random vertex $u \in V$
\State $C \gets \{u\}$ \Comment{output center set}
\State $d[v] \gets \infty \;\; \forall v \in V$ \Comment{global distance array}
\State \textproc{SSSP Routine}$(G, u, \infty, d)$ \Comment{initialize distances}

\While{$|C| < k$}
    \State $u \gets \arg\max_{v \in V} d[v]$
    \State $C \gets C \cup \{u\}$
    \State \textproc{SSSP Routine}$(G, u, \infty, d)$ \Comment{update distances}
\EndWhile

\State \textbf{return} $C$

\end{algorithmic}
\end{algorithm}
The algorithm invokes one SSSP computation per selected center, hence the running time of Gonzalez' algorithm on weighted graphs is $O(k(m + n\log n))$ using Dijkstra's algorithm. Since the algorithm always picks the exact farthest vertex, the randomized charging argument from the previous section does not apply. Consequently, the work and depth of the underlying SSSP computations in the parallel setting accumulate over the $k$ iterations. Replacing Lazy Dijkstra by Lazy $\Delta$-stepping therefore yields work $W = O(k(\ELDS+\PLDS))$ and depth $D = O(k\DLDS)$ for graphs with arbitrary positive edge weights. The implementation details are contained in~\cref{sec:GonzalezImp} of the appendix.

We remark that Forster, Nazari, Rao, and Skarlatos~\cite{ForsterNRS26} have recently observed that a dynamic single-source shortest path algorithm can be employed to amortize the cost of the shortest path computations performed for finding the farthest node in each iteration.
In particular, one can add a ``super-source'' to the graph and in each iteration of the Gonzalez algorithm insert an edge from this source to the farthest vertex and then update the distances.
This approach yields the state-of-the-art worst-case complexity for deterministic $k$-center by utilizing the dynamic approximate single-source shortest path algorithm of Górkiewicz and Karczmarz~\cite{GorkiewiczK25} supporting the insertion of edges incident on the source node, at the cost of an additional factor of $ 1 + \epsilon $ in the approximation guarantee.

We are however not aware of any practical implementations of (approximate) shortest path algorithms tailored to the edge-insertions setting.
A natural fully dynamic alternative (supporting also edge deletions, which is not necessary in this scenario), is the algorithm by Ramalingam and Reps \cite{RR} that has often proven to be one of the top contenders in terms of running time performance in experimental evaluations of dynamic shortest path algorithms \cite{Frigioni,Demetrescu}.
However, using this fully dynamic algorithm (and also others) would essentially be equivalent to the Lazy Dijkstra approach in this setting. 
For this reason, we did not explicitly include this approach to our study.
\subsection{$\alpha$-Approximate Gonzalez}
\label{sec:approx-gonzalez}

Another algorithm we implemented is the $\alpha$-approximate Gonzalez algorithm by Abboud et al \cite{Abboud}. Eppstein et al. independently studied approximate greedy permutations for graph metrics, providing another approach to approximate farthest-first traversal~\cite{Eppstein20}. Instead of picking the exact farthest vertex as the next center, the algorithm assigns each vertex $v$ to a bag depending on $d_v$ and selects a random vertex of the highest non-empty bag. The bag structure allows the algorithm to identify an approximately farthest vertex efficiently while supporting lazy distance updates. The size of the bags depends on $\alpha = 1 + \frac{\epsilon}{2}$. This affects the solution by a factor of $\alpha$, meaning that the radius output by this algorithm is a $(2+\epsilon)$-approximation for the above choice of $\alpha$.

\begin{algorithm}

\caption{\textsc{ApproximateGonzalez}$(G, k, \alpha)$}
\label{alg:approx-gonzalez}
\begin{algorithmic}[1]

\Require Graph $G=(V,E)$, number of centers $k$, approximation parameter $\alpha > 1$
\Ensure Center set $C$

\State $R \gets nW$
\State $B \gets \lceil \log_{\alpha} R \rceil$
\State $C \gets \emptyset$ \Comment{output center set}
\State $H_1,\ldots,H_{B-1} \gets \emptyset$, $H_B \gets V$ \Comment{distance bags}
\State $d[v] \gets \infty \;\; \forall v \in V$ \Comment{global distance array}
\State $b_{\max} \gets B$

\While{$|C| < k$}
    \While{$H_{b_{\max}} = \emptyset$}
        \State $b_{\max} \gets b_{\max} - 1$
    \EndWhile
    \State $u \gets$ random vertex from $H_{b_{\max}}$
    \State $C \gets C \cup \{u\}$
    \State $H_{b_{\max}} \gets H_{b_{\max}} \setminus \{u\}$
    \State \textproc{AugmentedSSSP}$(G, u, d, \alpha, H_1,\ldots,H_B)$
\EndWhile

\State \textbf{return} $C$

\end{algorithmic}
\end{algorithm}

The algorithm relies on an augmented SSSP routine that updates both the distance
array and the bags containing the vertices affected. As an example, the
augmented version of Dijkstra's algorithm is contained in~\cref{sec:ApproximateGonzalezImp} of the appendix. Abboud et al. show that this algorithm has running time $O\left((n \log^2 n + m \log n)\frac{\log n}{\epsilon}\right)$ in expectation. The independence of $k$ again follows from the randomized halving argument: each vertex has its distance to the current set of centers improved only $O(\log n)$ times per bag in expectation~\cite{Abboud}. The additional bag maintenance performed after each successful distance decrease requires only constant work per successful distance decrease and therefore does not affect the asymptotic running time. Replacing Lazy Dijkstra with Lazy $\Delta$-stepping changes only the cost of the underlying SSSP computations while leaving the randomized charging argument unchanged. Consequently, in the parallel setting we have expected work $W = O(\ELDS\frac{\log^2 n}{\epsilon} + k\PLDS)$ and depth $D = O(k\DLDS).$ 
During the augmented SSSP calls, vertices that have their distance to the current set of centers improved need to be moved to the correct bag. Note that all vertices are initialized in the last bag (\Cref{alg:approx-gonzalez}, line~4). As a consequence, the first time a vertex $v$ has its distance relaxed, it is removed from $H_B$, not from a bag determined by $\log_\alpha d[y]$ as stated in \Cref{alg:augmented-dijkstra}, line~8. The implementation details are contained in~\cref{sec:ApproximateGonzalezImp} of the appendix.
\section{Experimental Setup}
\label{sec:setup}
\subsection{General Setup}

All algorithms presented in Sections~\ref{sec:mis} and \ref{sec:farthest} were implemented in C++20 as part of the Graph-Based Benchmark Suite \cite{GBBS}. It uses a shared-memory approach to parallel graph processing in which the entire graph is stored in main memory of a single multi-core machine. The code of the library exploits nested parallelism using scheduler-agnostic parallel primitives, such as fork-join and parallel-for loops. The source code of all algorithms evaluated in this paper is publicly available on Zenodo.\footnote{\url{https://doi.org/10.5281/zenodo.23433263}} The code was compiled using g++ 12 through Bazel using optimization level \texttt{-O2} and architecture-specific optimization enabled through \texttt{-march=native}. All experiments were conducted under Debian 12 bookworm on a server equipped with two AMD EPYC 7763 64-core processors and approximately $4$ TB of RAM available, accessed through a virtual machine. Unless stated otherwise, each experiment was repeated three times and the median running time is reported. The timeout for the experimental sweep across all instances was set to three hours per run. \\

We evaluated the algorithms on both weighted and unweighted real-world graph instances originating from various graph collections such as road networks, social networks and other graph types such as graphs representing movie ratings. These datasets have been used extensively in other graph algorithm studies \cite{Laxman}, \cite{Laxman2}, \cite{Pathak}, \cite{Henzinger}. They cover a wide range of graph sizes ranging from approximately one million vertices and three million edges up to over $60$ million vertices and $1.8$ billion edges. In addition, we generated sets of small and large weighted and unweighted synthetic graphs to complement our experiments.\\

For each algorithm we focus mainly on evaluating their running time and the quality of their output solution. As $k$-center on graphs is an NP-hard problem, the solution quality is determined relative to the output of the other algorithms. In addition, we also measure how well the algorithms scale with parallelism. Each algorithm was evaluated in both its sequential and parallel variants on all instances. For shortest-path computations on weighted graphs, the sequential implementations use Lazy Dijkstra or Lazy Capped Dijkstra, whereas the parallel implementations employ Lazy $\Delta$-stepping or Lazy Capped $\Delta$-stepping depending on the algorithm. On unweighted graphs, Lazy BFS, Lazy Capped BFS and their parallel variants are used analogously. Throughout the experimental evaluation we consider values of $k$ up to approximately $\sqrt{n}$. This regime already represents a large number of centers for the considered graph sizes and covers the range typically relevant in practical facility location and clustering applications.
\subsection{Datasets}

For all datasets, we performed a preprocessing step which removes isolated vertices, multi-edges and self-loops. Most instances are already simple graphs such as the instances from SNAP \cite{SNAP} and for those that are not, this type of preprocessing is also a common practice in the literature \cite{Laxman}. For $k$-center in particular, each isolated vertex necessarily forms a singleton cluster and therefore does not affect the remaining optimization problem. In the case of weighted graphs, multi-edges were removed by keeping the edge with the lowest weight. The numbers on graph sizes shown in the tables below were taken after this clean-up procedure. Hence, the numbers may slightly differ from those given by the original sources.

\subsubsection{Synthetic Graphs}

We generated sets of small and large random graphs according to the Erdős–Rényi model for both the weighted and the unweighted setting using our own graph generator provided in the Zenodo archive. The small graphs have $10^5$ vertices, while the large graphs have $10^7$ vertices. The density of these graphs is $\rho = \frac{m}{n}$ for $\rho \in \{2,4,8,16,32,64\}$. For the small instances, we generated $10$ graphs per density. In total, this yields $120$ small and $12$ large instances. The integer weights on the weighted graphs were taken uniformly at random from $[1, 10000]$. The idea was to cover a wide variety of topologies, including sparse graphs with high diameter and dense graphs with low diameter.

\subsubsection{Social Networks}

The first group of real world graphs that we evaluated our algorithms on represent social networks. These graphs were obtained from the Stanford Network Analysis Platform (SNAP) \cite{SNAP}. The main properties of social networks are that they are unweighted, have high average degree and as a result, tend to have low diameter.
\begin{table}[H]
\setlength{\tabcolsep}{5pt}
\begin{tabularx}{\columnwidth}{lrrr}
\toprule
Graph & $n$ & $m$ & $\rho$ \\
\midrule
DBLP        & 317,080    & 1,049,866     & 3.31  \\
YouTube     & 1,134,890  & 2,987,624     & 2.63  \\
LiveJournal & 3,997,962  & 34,681,189    & 8.67  \\
Orkut       & 3,072,441  & 117,185,083   & 38.14 \\
Twitter     & 41,652,230 & 1,202,513,046 & 28.87 \\
Friendster  & 65,608,366 & 1,806,067,135 & 27.53 \\
\bottomrule
\end{tabularx}
\caption{Characteristics of the social network benchmark instances used in our experiments after preprocessing.}
\label{tab:social-networks}
\end{table}

\subsubsection{Road Networks}

We also ran our experiments on the two largest road network graphs available from the $9$th DIMACS challenge \cite{DIMACS}. These graphs are very sparse ($\rho < 2$) and have high diameter as the weights represent physical distances across the US. 
\begin{table}[H]
\setlength{\tabcolsep}{5pt}
\begin{tabularx}{\columnwidth}{lrrr}
\toprule
Graph & $n$ & $m$ & $\rho$ \\
\midrule
USA-central & 14,081,816 & 16,933,413 & 1.2 \\
USA-full    & 23,947,347 & 28,854,312 & 1.2  \\
\bottomrule
\end{tabularx}
\caption{Road network instances from the 9th DIMACS shortest path implementation challenge after preprocessing.}
\label{tab:road_networks}
\end{table}

\subsubsection{Other Networks}

We complement our experiments with three additional rating networks: libimseti from the Network Repository \cite{NR}, MovieLens from GroupLens \cite{Harper}, and Yahoo! Song from KONECT \cite{KONECT}. These graphs represent movie or song ratings and fill the gap between the dense unweighted social networks and the sparse weighted road networks. 
\begin{table}[H]
\setlength{\tabcolsep}{5pt}
\begin{tabularx}{\columnwidth}{lrrr}
\toprule
Graph & $n$ & $m$ & $\rho$ \\
\midrule
libimseti   &   220,970 &  17,233,144 &  77.99   \\
MovieLens   &   414,214 &  33,832,162 &  81.67 \\
Yahoo! Song & 1,625,951 & 256,804,235 & 157.94  \\
\bottomrule
\end{tabularx}
\caption{Rating graph instances after preprocessing.}
\label{tab:other_networks}
\end{table}
Since a higher rating in the context of these graphs is preferable to a lower rating, we transformed the weights accordingly to fit our shortest-path regime: For example, if the highest rating on a specific instance is $5$, then this gets mapped to the best possible edge weight of $1$ under the shortest-path metric. The common property of these graphs is that they are very dense and overall have a low diameter. However, the diameter is still significantly higher than for unweighted instances of similar size.
Before the actual experimental sweep across all instances, we performed pre-tests on various instances regarding the choice of algorithm-specific parameters. These are contained in~\cref{sec:parameters} of the appendix. Further, we also performed parallel scaling experiments on a subset of graphs to observe initial scaling trends and ideally find the best performing thread count for the full experimental sweep. These experiments are discussed in the following section.
\section{Experimental Results}
\label{sec:results}
In this section, we show and discuss our most important results, beginning with the parallel scaling experiments.
\subsection{Parallel Scaling}
\label{sec:parallelscaling}
We conducted experiments to investigate how well our algorithms scale with parallelism. The following plots show the performance of all algorithms on a variety of test instances, comparing their single-core variants with parallel executions using increasing numbers of threads. The results for small and large synthetic graphs are contained in~\cref{sec:parallelsmall} and~\cref{sec:parallellarge} respectively of the appendix.
\subsubsection{Parallelism on Social Networks}
\label{sec:parallelreal}
\begin{figure}
\centering
\begin{minipage}{0.88\linewidth}
\centering
\includegraphics[width=0.9\linewidth]{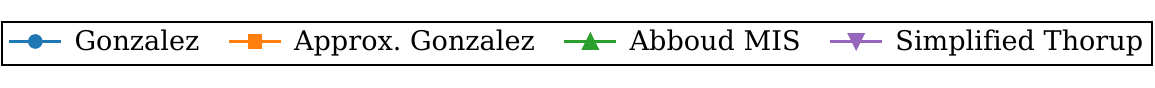}\\

\begin{subfigure}[t]{0.48\linewidth}
\centering
\includegraphics[width=\linewidth]{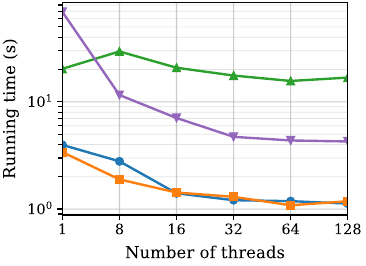}
\caption{$k=200$, LiveJournal}
\end{subfigure}
\hfill
\begin{subfigure}[t]{0.48\linewidth}
\centering
\includegraphics[width=\linewidth]{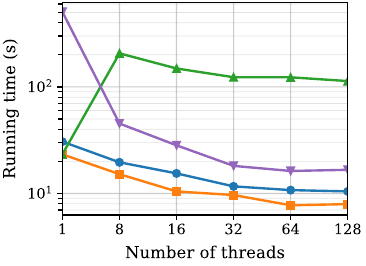}
\caption{$k=2000$, LiveJournal}
\end{subfigure}
\\[-0.15em]

\begin{subfigure}[t]{0.48\linewidth}
\centering
\includegraphics[width=\linewidth]{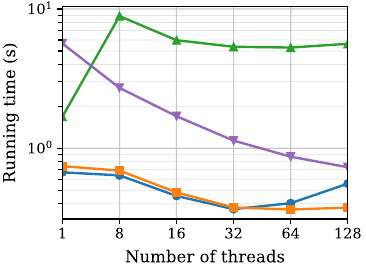}
\caption{$k=200$, YouTube}
\end{subfigure}
\hfill
\begin{subfigure}[t]{0.48\linewidth}
\centering
\includegraphics[width=\linewidth]{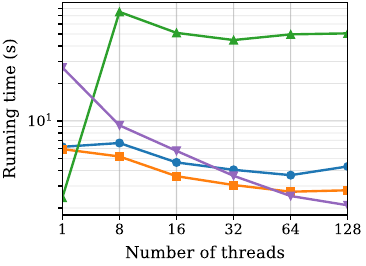}
\caption{$k=2000$, YouTube}
\end{subfigure}
\\[-0.15em]

\begin{subfigure}[t]{0.48\linewidth}
\centering
\includegraphics[width=\linewidth]{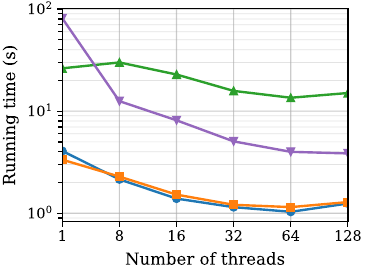}
\caption{$k=200$, Orkut}
\end{subfigure}
\hfill
\begin{subfigure}[t]{0.48\linewidth}
\centering
\includegraphics[width=\linewidth]{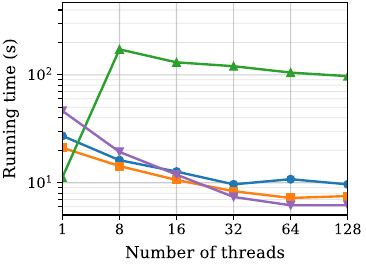}
\caption{$k=2000$, Orkut}
\end{subfigure}
\\[-0.15em]

\caption{Parallel scaling on three social networks. The leftmost point represents single-core execution.}

\label{fig:parallel-scaling-social-networks}
\end{minipage}
\end{figure}

We also evaluated the effects of parallelism on three graphs from our social-network dataset. Since these graphs are unweighted, relatively dense and have low diameter, they complement the previously shown synthetic instances. The results are consistent with our earlier observations. The farthest-point-based algorithms and the simplified Thorup algorithm benefit from parallelism up to approximately $64$ threads, reducing their running times by factors between roughly $2$ and $5$. For Abboud MIS, the effect depends on both $k$ and the density of the graph. YouTube is the sparsest graph in this experiment, and its running time is negatively affected by parallelism even for the smaller value of $k$. LiveJournal and especially Orkut are denser, and parallelism improves the running time by up to a factor of approximately $2$, as shown in Figure~\ref{fig:parallel-scaling-social-networks}, subfigure~(e). For values of $k$ approaching $\sqrt{n}$, however, the parallel execution becomes substantially slower for the reasons discussed above. 

\begin{figure}
\centering
\begin{minipage}{0.88\linewidth}
\centering
\includegraphics[width=0.98\linewidth]{Plots/parallel_scaling_social_networks_legend.pdf}\\

\begin{subfigure}[t]{0.48\linewidth}
\centering
\includegraphics[width=\linewidth]{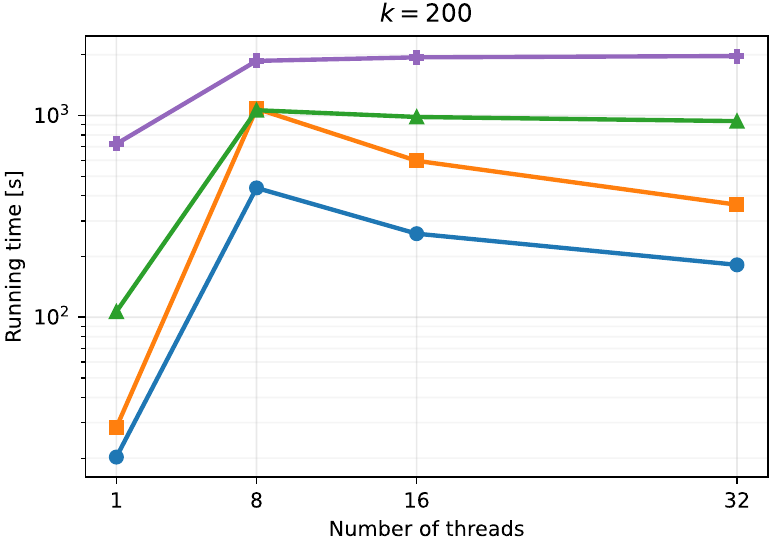}
\caption{$k=200$}
\end{subfigure}
\hfill
\begin{subfigure}[t]{0.48\linewidth}
\centering
\includegraphics[width=\linewidth]{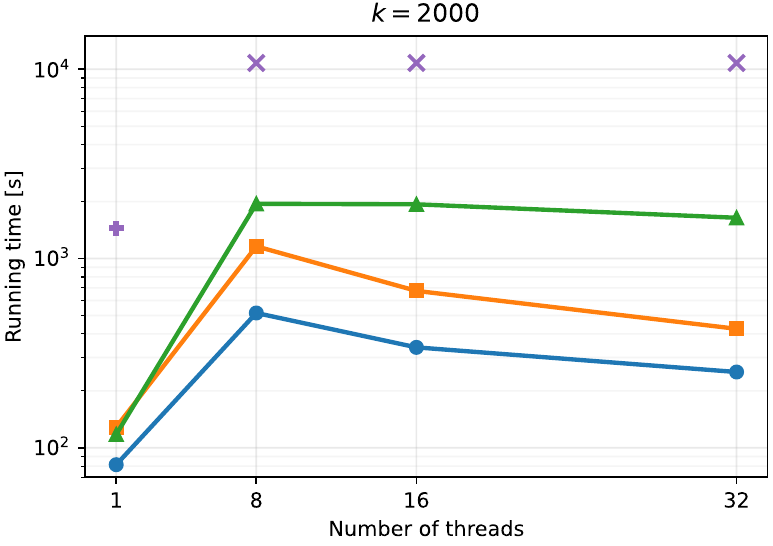}
\caption{$k=2000$}
\end{subfigure}

\caption{Parallel scaling on the USA-central road network. The leftmost point represents the single-core execution. An $\times$ denotes a timeout.}
\label{fig:parallel-scaling-road-network}
\end{minipage}
\end{figure}

Finally, we also examined parallel scaling on a road-network graph. This graph has very low density and high diameter. As shown in Figure~\ref{fig:parallel-scaling-road-network}, no algorithm benefits from parallelism. Although increasing the thread count improves the running time of most parallel executions compared with $8$ threads, these improvements are not sufficient to offset the overhead relative to the single-core implementations.
For the road network, we restrict the dedicated scaling experiment to at most $32$ threads in order to characterize the initial scaling behavior. The full experimental sweep additionally contains executions with $64$ threads for different values of $k$. Although these results are
not directly comparable point by point, they confirm the same general trend: single-core execution remains faster for all algorithms.
To summarize, these experiments have shown that if instances are large enough and have sufficient density, most algorithms significantly benefit from parallelism up to around $64$ threads which is why we chose to use this thread count for our full experimental sweep over all instances for the parallel executions.\footnote{Keep in mind, we are accessing the machine our experiments were conducted on a virtual machine, which might influence these results} Abboud MIS does not benefit from parallelism once $k$ becomes too large. If graphs are too small, parallelism does not pay off for any algorithm. On sparse graphs with high diameter all algorithms generally perform better in single-core mode.
\subsection{Running Time and Solution Quality Results}
\label{sec:RTandSQresults}
In this subsection, we show and discuss the most important results of our full experimental sweep across all instances. The results for synthetic and the real world rating graphs are contained in~\cref{sec:ResultsSmall},~\ref{sec:ResultsLarge} and~\ref{sec:ResultsOther} in the appendix.
\subsubsection{Social Networks}
\begin{figure}
\centering
\begin{minipage}{0.88\linewidth}
\centering

\includegraphics[width=0.90\linewidth]{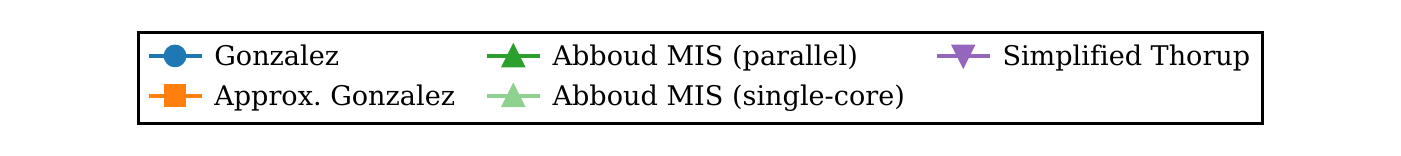}\\[0.1em]

\begin{subfigure}[t]{0.48\linewidth}
\centering
\includegraphics[width=\linewidth]{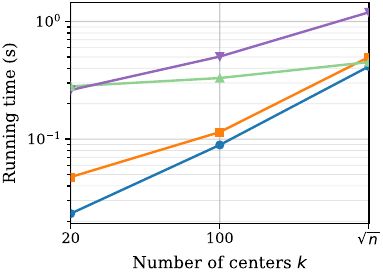}
\caption{DBLP}
\end{subfigure}
\hfill
\begin{subfigure}[t]{0.48\linewidth}
\centering
\includegraphics[width=\linewidth]{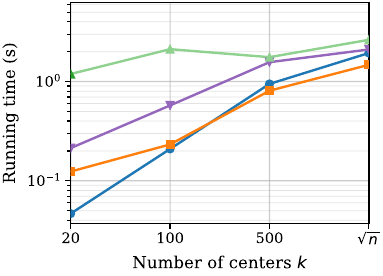}
\caption{YouTube}
\end{subfigure}

\par\vspace{-0.15em}

\begin{subfigure}[t]{0.48\linewidth}
\centering
\includegraphics[width=\linewidth]{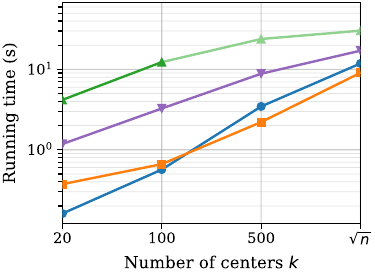}
\caption{LiveJournal}
\end{subfigure}
\hfill
\begin{subfigure}[t]{0.48\linewidth}
\centering
\includegraphics[width=\linewidth]{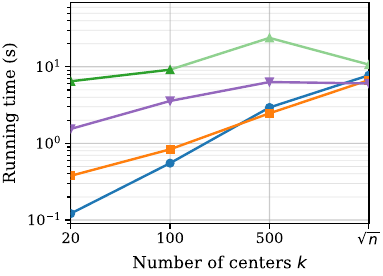}
\caption{Orkut}
\end{subfigure}

\par\vspace{-0.15em}

\begin{subfigure}[t]{0.48\linewidth}
\centering
\includegraphics[width=\linewidth]{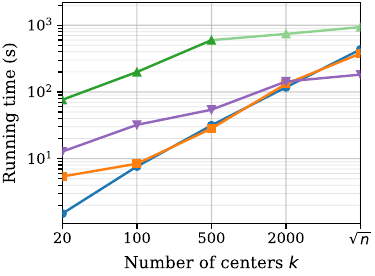}
\caption{Twitter}
\end{subfigure}
\hfill
\begin{subfigure}[t]{0.48\linewidth}
\centering
\includegraphics[width=\linewidth]{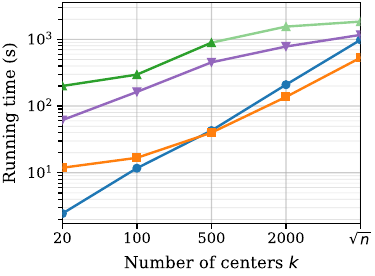}
\caption{Friendster}
\end{subfigure}

\caption{Running times on the social-network instances. All algorithms use parallel execution with $64$ threads except for Abboud MIS, where each point shows the faster result of the parallel (dark) and single-core (light) runs.}
\label{fig:social-full-sweep-runtime}

\end{minipage}
\end{figure}

On the social networks, the MIS based approaches are generally slower for lower values of $k$ than the farthest-point based algorithms but as $k$ approaches $\sqrt{n}$, they catch up. In particular, the Simplified Thorup algorithm outperformed the farthest-point algorithms on Orkut and Twitter. However, it did not outperform them on Friendster, which we believe is due to the fact that the reported radii of all algorithms on this graph was considerably higher. This suggests that this graph has a substantially larger effective diameter, increasing the number of feasibility checks performed by MIS-based algorithms. The quality of solution plots for these networks are contained in~\cref{sec:socQual} in the appendix.
\subsubsection{Road Networks}

\begin{figure}[t]
\centering
\includegraphics[width=0.8\linewidth]{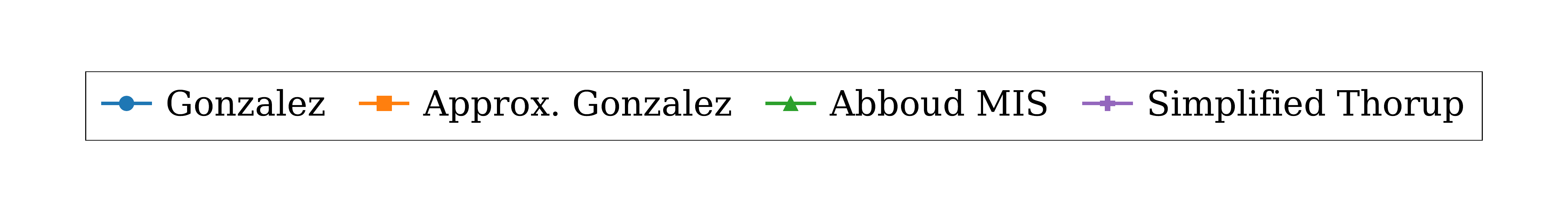}
\vspace{0.5em}
\begin{subfigure}[t]{0.48\linewidth}
\centering
\includegraphics[width=\linewidth]{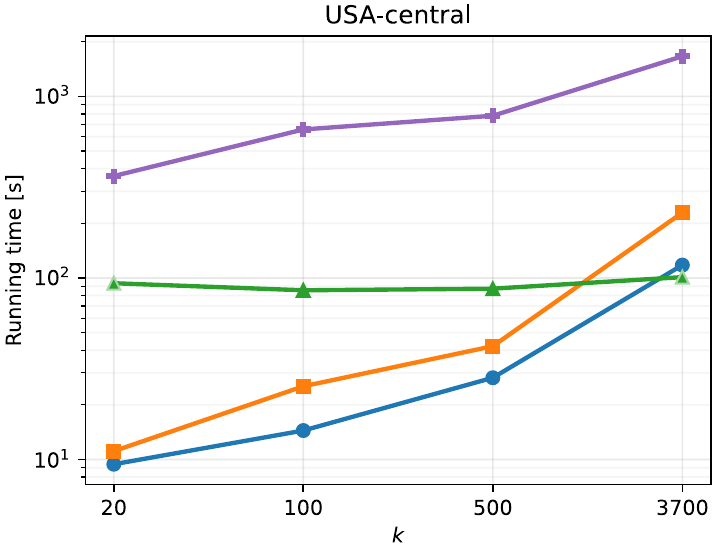}
\caption{USA-central}
\end{subfigure}
\hfill
\begin{subfigure}[t]{0.48\linewidth}
\centering
\includegraphics[width=\linewidth]{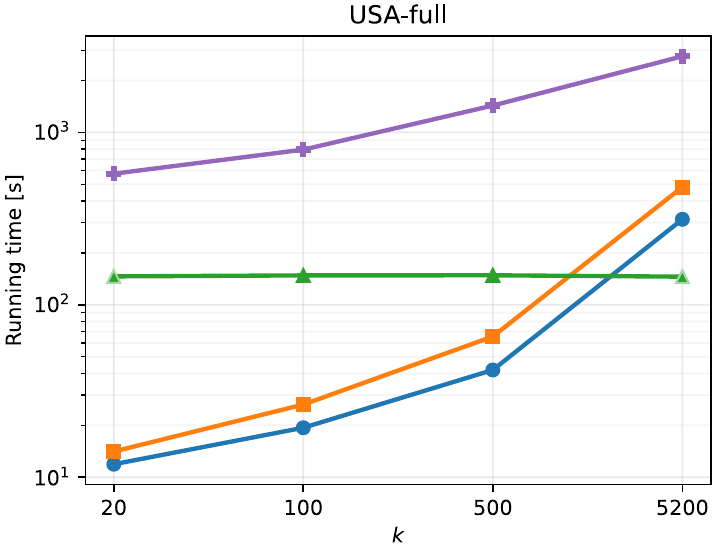}
\caption{USA-full}
\end{subfigure}
\caption{Running times on the two road-network instances. All algorithms use single-core execution.}
\label{fig:road-full-sweep-runtime}
\end{figure}

The running times on the road-network graphs confirm the trends observed in the preceding scaling experiments. All algorithms achieve their best performance in single-core mode. The two farthest-point-based algorithms perform best for smaller values of $k$. As $k$ increases, however, Abboud MIS gradually closes the gap and eventually outperforms them on both instances. Simplified Thorup on the other hand remains substantially slower throughout. This is consistent with our earlier observation that weighted graphs with high diameter require many feasibility checks, making the additional phases of Simplified Thorup dominate the running time.

\begin{figure}[t]
\centering
\includegraphics[width=0.8\linewidth]{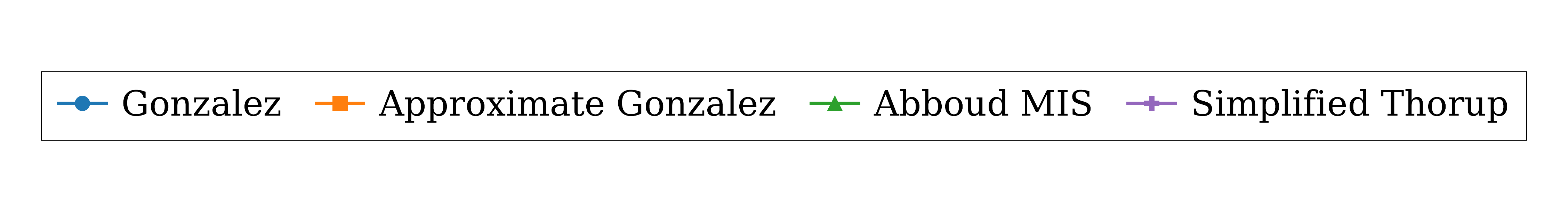}
\vspace{0.5em}
\begin{subfigure}[t]{0.48\linewidth}
    \centering
    \includegraphics[width=\linewidth]{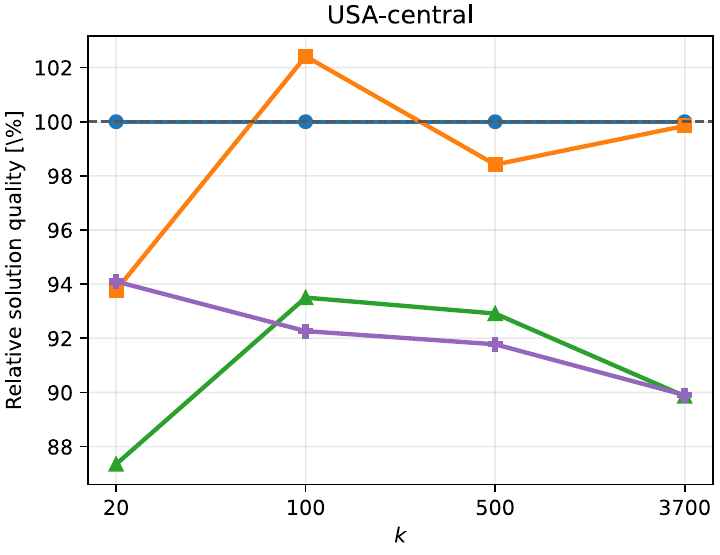}
    \caption{USA-central}
\end{subfigure}
\hfill
\begin{subfigure}[t]{0.48\linewidth}
    \centering
    \includegraphics[width=\linewidth]{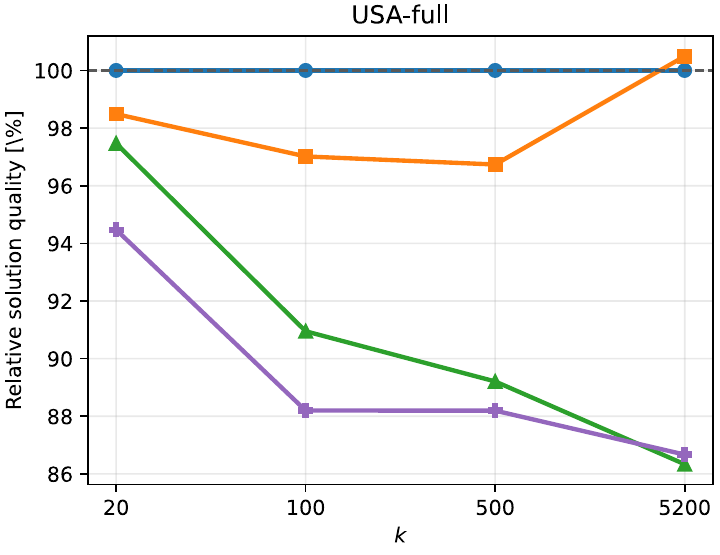}
    \caption{USA-full}
\end{subfigure}
\caption{Solution quality on the two road-network instances. The $k$-center radius returned by each algorithm is reported relative to the one returned by Gonzalez for the same graph and value of $k$.}
\label{fig:road-full-sweep-quality}
\end{figure}

In terms of solution quality, the MIS-based algorithms clearly emerge as the strongest performers on the road networks. Across all tested values of $k$, they consistently produce solutions that are approximately $10\%$ to $12\%$ better than the Gonzalez baseline. Confirming the parameter study, Approximate Gonzalez also generally outperforms vanilla Gonzalez on these instances. This suggests that, on sparse networks with high diameter, introducing randomization into the center selection process can lead to better practical solutions than always choosing the exact farthest vertex.
\section{Conclusion}
\label{sec:conclusion}

In this paper, we presented an experimental study of static algorithms for the $k$-center problem on graphs. We implemented four algorithms, including both classical approaches such as Gonzalez, an algorithm by Thorup as well as more recent ones like an improved variant of Hochbaum \& Shmoys by Abboud et al., and an $\alpha$-approximate variant of Gonzalez. All algorithms were implemented as part of the Graph-Based Benchmark Suite and share the same optimized shortest-path routines, allowing for a fair comparison of their practical performance.

Our experiments demonstrate that no single algorithm is universally superior. The farthest-point-based algorithms consistently achieve the best running times on weighted graphs unless $k$ is very large and the instances are very sparse. 
Meanwhile, the MIS-based approaches become competitive or even outperform them on such sparse weighted graphs and dense unweighted graphs for larger values of $k$. It is worth mentioning that in practice, the MIS-based algorithms' running time performance is heavily dependent on the realization of the binary search. In practical applications with knowledge about the dataset, a heuristic approach for the initial radius guess might already yield significant improvements. The study further shows that graph density, edge weights, diameter, and the choice of $k$ all have a substantial influence on practical performance. Likewise, parallel execution is highly beneficial on sufficiently large instances as long as they are not particularly sparse, but the achievable speedups strongly depend on the underlying algorithm and graph topology. In terms of solution quality, all algorithms reported similar solutions in absolute terms on unweighted instances. On the weighted instances however, the MIS algorithms generally reported better solutions than the farthest-point algorithms with differences in relative solution quality ranging from approximately $10\%$ to $12\%$ on road networks.

Overall, our study provides both practical implementations and an extensive experimental comparison of modern graph-based $k$-center algorithms. We hope that the implementation techniques and empirical insights presented in this paper will facilitate future work on scalable clustering algorithms for large-scale graphs.
\newpage

\appendix
\section{Appendix} 
\subsection{$\Delta$-stepping}
\label{sec:ds}

In this section we explain the simple parallel $\Delta$-stepping algorithm used by our algorithms to perform parallel shortest path computations on weighted graphs \cite{Meyer}. As mentioned previously, it receives two modifications: Lazy updates and an optional distance cap. Before discussing the effects of these modifications, let us first recall the running time of $\Delta$-stepping without modifications. The parallel bounds stated below are based on the analysis of Meyer and Sanders~\cite{Meyer}, which assumes an arbitrary CRCW PRAM. In this model, concurrent shared-memory accesses take constant time and an arbitrary write succeeds in the case of a conflict. For $\Delta$-stepping, competing relaxation requests for the same vertex are first
combined by selecting the minimum proposed distance.

\begin{algorithm}
\caption{\textsc{$\Delta$-stepping}$(G, s, d)$}
\label{alg:lds}
\begin{algorithmic}[1] 

\Require Graph $G=(V,E)$, starting vertex $s$, array of distances $d$
\Ensure $d[v]$ stores the distance from $s$ to $v$ for every reachable vertex $v \in V$

\State $d[v] \leftarrow \infty \: \forall v \in V$
\State $d[s] \leftarrow 0$
\State $B_0 \leftarrow$ bucket with single vertex $s$
\For{$i$ from $0$ to $\frac{L}{\Delta}$} 
    \State $R := \emptyset$ \Comment{Set of vertices removed from the current bucket}
    \While{$B_i \neq \emptyset$}
        \State $S \leftarrow B_i$ \Comment{Snapshot of current bucket}
        \State $B_i \leftarrow \emptyset$
        \State $R \leftarrow R \cup S$
        \ParForEach{light edge $(u,v) \in E$ with $u \in S$ and $w(u,v) \leq \Delta$}
            \If{$d_v > d_u + w(u,v)$}
                \State $d_v \leftarrow d_u + w(u,v)$
                \State Move $v$ into new bucket if necessary
                    \Comment{$v$ may (re)enter $B_i$}
            \EndIf
        \EndParForEach
    \EndWhile
    \ParForEach{heavy edge $(u,v) \in E$ with $u \in R$ and $w(u,v) > \Delta$}
        \If{$d_v > d_u + w(u,v)$}
            \State $d_v \leftarrow d_u + w(u,v)$
            \State Move $v$ into new bucket if necessary
                \Comment{$v$ may not (re)enter $B_i$}
        \EndIf
    \EndParForEach
\EndFor

\end{algorithmic}
\end{algorithm}

\begin{theorem}[Meyer and Sanders~\cite{Meyer}]
\label{thm:delta-stepping}

On graphs with arbitrary positive edge weights, simple parallel $\Delta$-stepping has work

\[
W =
O\!\left(
n + m + n_\Delta + m_\Delta
+ \frac{L}{\Delta}\,
\degmax\,\ell_\Delta \log n
\right)
\]

and depth

\[
D =
O\!\left(
\frac{L}{\Delta}\,
\degmax\,\ell_\Delta \log n
\right)
\]

with high probability on an arbitrary CRCW PRAM.

\end{theorem}
 
 In \Cref{thm:delta-stepping}, $\Delta$ denotes the bucket width, $n_\Delta$ and $m_\Delta$ are graph-dependent parameters introduced by Meyer and Sanders which upper-bound the total number of reinsertions of vertices into the current bucket and the re-relaxations of their incident light edges, $L$ denotes the eccentricity of the starting vertex and thus $\frac{L}{\Delta}$ is the number of buckets that are processed. Further, $\ell_\Delta$ corresponds to one plus the maximum number of edges that occur on any traversed path with total weight at most $\Delta$. In other words, $\ell_\Delta$ is a parameter which upper bounds the number of phases needed to process each nonempty bucket. Finally, $\degmax$ denotes the maximum vertex degree of the graph. In the simple parallel implementation of Meyer and Sanders, the $\degmax\log n$ factor results from the randomized distribution and load balancing of relaxation requests among the processors. Their analysis bounds the load imbalance by $O(\degmax\log n)$ per phase with high probability. Further, Meyer and Sanders show that under the additional assumption that edge weights are random (i.e., independent and uniformly distributed in $[0, 1]$), Meyer and Sanders prove that choosing $\Delta = \Theta(\frac{1}{\degmax})$ yields that $\mathbb{E}[n_\Delta]=O(n)$ and $\mathbb{E}[m_\Delta]=O(n)$, and that both bounds also hold with high probability. This leads to simplified bounds of $O(n + m + \degmax^2 \cdot L \cdot \log^2 n)$ expected work and $O(\degmax^2 \cdot L \cdot \log^2 n)$ depth w.h.p. \cite{Meyer}. Since our analysis does not rely on these assumptions, we use the general work-depth bounds stated above throughout the remainder of this paper.

\begin{algorithm}
\caption{\textsc{LazyCapped$\Delta$-stepping}$(G, s, r, d)$}
\label{alg:mds}

\begin{algorithmic}[1] 

\Require Graph $G=(V,E)$, starting vertex $s$, array of distances $d$, radius $r$
\Ensure $d[v]$ stores the smallest known distance to $v$ for every vertex reachable within distance $r$

\State $d[s] \leftarrow 0$
\State $B_0 \leftarrow$ bucket with single vertex $s$
\For{$i$ from $0$ to $\ceil{\frac{r}{\Delta}}$} 
    \State $R := \emptyset$ \Comment{Set of vertices removed from the current bucket}
    \While{$B_i \neq \emptyset$}
        \State $S \leftarrow B_i$ \Comment{Snapshot of current bucket}
        \State $B_i \leftarrow \emptyset$
        \State $R \leftarrow R \cup S$
        \ParForEach{light edge $(u,v) \in E$ with $u \in S$, $w(u,v) \leq \Delta$ and $d_u + w(u,v) \leq r$}
            \If{$d_v > d_u + w(u,v)$}
                \State $d_v \leftarrow d_u + w(u,v)$
                \State Move $v$ into new bucket if necessary
                    \Comment{$v$ may (re)enter $B_i$}
            \EndIf
        \EndParForEach
    \EndWhile
    \ParForEach{heavy edge $(u,v) \in E$ with $u \in R$, $w(u,v) > \Delta$ and $d_u + w(u,v) \leq r$}
        \If{$d_v > d_u + w(u,v)$}
            \State $d_v \leftarrow d_u + w(u,v)$
            \State Move $v$ into new bucket if necessary
                \Comment{$v$ may not (re)enter $B_i$}
        \EndIf
    \EndParForEach
\EndFor

\end{algorithmic}
\end{algorithm}

Now let us take a look at the slightly modified version of simple parallel $\Delta$-stepping. As shown in \Cref{alg:mds}, two slight modifications have been made to the algorithm. The first one is lazy updates: The algorithm does not reset the distance array $d$ and thus vertices are only updated if their new distance improves upon the best known previous distance. For a single invocation, this modification can only reduce the number of successful relaxations and does not worsen the asymptotic work or depth bound.\\

The second modification is a distance cap $r$, which ensures that an edge $(u,v)$ is only relaxed if $d_u + w(u,v) \leq r$. This modification has effects on the work and depth of a single run. The effects are as follows:

\begin{enumerate}
\item {The explored number of vertices and edges is reduced from $n$, $m$ to $n_r$, $m_r$}.
\item {Similarly, the number of reinsertions and re-relaxations is reduced from $n_\Delta$, $m_\Delta$ to $n_{r, \Delta}$, $m_{r, \Delta}$}.
\item {The number of buckets is reduced from $\frac{L}{\Delta}$} to $\ceil{\frac{r}{\Delta}}$.
\end{enumerate}

Instead of all vertices and edges, we now have to only consider them when they are within $r$. Since buckets are only created if at least one vertex is inserted, the number of buckets is also reduced. The following lemma gives the corresponding work and depth bounds for Lazy Capped $\Delta$-stepping.

\begin{lemma}
\label{lem:lcds}

Lazy Capped $\Delta$-stepping on graphs with arbitrary positive edge
weights has work

\[
W =
O\!\left(
n_r+m_r+n_{r,\Delta}+m_{r,\Delta}
+\left\lceil\frac{r}{\Delta}\right\rceil
\degmax\,\ell_\Delta\log n
\right)
\]

and depth

\[
D =
O\!\left(
\left\lceil\frac{r}{\Delta}\right\rceil
\degmax\,\ell_\Delta\log n
\right).
\]

\end{lemma}

\begin{proof}
Lazy updates never increase the number of successful relaxations within a
single execution and therefore do not worsen the asymptotic work or
depth. Introducing a distance cap restricts the computation to vertices
and edges within distance $r$ from the source. Consequently, the
explored graph is reduced from $(n,m)$ to $(n_r,m_r)$, the numbers of
reinsertions and re-relaxations become $n_{r,\Delta}$ and
$m_{r,\Delta}$, and the number of processed buckets decreases from
$\frac{L}{\Delta}$ to
$\left\lceil\frac{r}{\Delta}\right\rceil$.
Plugging these variables into the original work and depth bounds of
Meyer and Sanders yields the result.
\end{proof}

Those modifications can be applied to the other SSSP routines with analogous effects on their asymptotic running time. In our implementations, the underlying SSSP routines have the following modifications based on the calling $k$-center algorithm:

\begin{enumerate}
    \item MIS based algorithms: SSSP computations use both lazy updates and a distance cap. 
    \item Farthest point selection based algorithms: SSSP computations only use lazy updates.
\end{enumerate}

 Hence, for the remainder of this paper we use the prefixes \textbf{Lazy} or \textbf{Lazy Capped} to indicate whether an SSSP routine employs only lazy updates or both lazy updates and a distance cap, respectively, for example Lazy Dijkstra, Lazy $\Delta$-stepping, Lazy Capped Dijkstra, or Lazy Capped $\Delta$-stepping. Further, we decompose the work of $\Delta$-stepping into two parts: the exploration work, corresponding to vertex visits, edge relaxations, reinsertions, and re-relaxations, and the phase overhead, corresponding to the synchronization cost incurred in each processing phase. We use the following notation throughout the remainder of the paper:
 
\begin{itemize}

\item For Lazy $\Delta$-stepping, let
\begin{align*}
&\ELDS := n+m+n_\Delta+m_\Delta, \\
&\PLDS := \frac{L}{\Delta}\degmax\,\ell_\Delta\log n.
\end{align*}

\item For Lazy Capped $\Delta$-stepping, let
\begin{align*}
&\ELCDS := n_r+m_r+n_{r,\Delta}+m_{r,\Delta}, \\
&\PLCDS := \left\lceil\frac{r}{\Delta}\right\rceil\degmax\,\ell_\Delta\log n.
\end{align*}

\item Consequently,
\begin{align*}
&\WLDS = O(\ELDS+\PLDS), \\
&\DLDS = O(\PLDS),
\end{align*}
and
\begin{align*}
&\WLCDS = O(\ELCDS+\PLCDS), \\
&\DLCDS = O(\PLCDS).
\end{align*}

\end{itemize}
\subsection{Omitted Proof of Section~\ref{sec:abboud}}
\label{sec:abboudProof}
\begin{proof}[Proof of running time]
    Consider the $i$-th iteration of the while loop. For each vertex $u \in V$, let $Y_i^u$ be the indicator random variable that denotes whether the global distance of $u$ is updated in this round. Keep in mind we use lazy updates, i.e. a vertex $v$ is only inserted into a bucket if its best known previously computed distance has been updated. \\
    
    For the sake of simplicity, before accounting for the remaining work performed by Lazy Capped $\Delta$-stepping, we first analyze the exploration work by isolating the contribution of successful distance decreases analogously to the analysis by Abboud et al. \cite{Abboud}.
    In the following, we temporarily ignore the distance cap (which can only reduce the explored subgraph) as well as the distinction between light and heavy edges, bucket reinsertions, and re-relaxations within a bucket, instead charging work to vertices whenever their distance is strictly decreased. With that in mind, the total work of \Cref{alg:mis} attributable to successful distance decreases is at most
    \begin{align*}
    W' \;=\;O\!\left(\underbrace{\sum_{u \in V} \deg(u)\cdot \sum_{i=1}^{n}Y_i^u}_{\scriptsize\text{edge relaxations / scans}}\;+\;\underbrace{ \sum_{u \in V}\sum_{i=1}^{n} Y_i^u}_{\scriptsize\substack{\text{bucket operations (}O(1) \\ \text{ work per decrease)}}}\right).
    \end{align*}
    
    Here we count bucket insertions/moves as $O(1)$ work. Other bucket operations such as contention resolution is accounted for separately below. In other words, every time a vertex $u$'s distance is improved, it is inserted into a bucket which means it later has to be processed, i.e. its incident edges are scanned, relaxations and the corresponding bucket operations are performed. Now let $d_i^u$ and $T_i$ denote the distance $d_u$ and the set $T$ before the $i$-th iteration of \Cref{alg:mis} respectively. Additionally, let $W_i^u$ denote the set of all vertices in $T_i$ whose shortest path distance to $u$ is less than $d_i^u$, i.e $W_i^u := \{w \in T_i \mid d_G(w,u) < d_i^u\}$. That is to say that $W_i^u$ is the set of vertices which would improve the distance of $u$, if one of them were to be selected as the next center. Since $d_i^u$ is nonincreasing, we have that $W_1^u \supseteq W_2^u \supseteq ... \supseteq W_n^u$. Moreover, $Y_i^u = 0$ whenever $W_i^u = \emptyset$.\\ 
    
    Now suppose vertex $v$ is selected as the next center during iteration $i$. If $v$ is not an element of $W_i^u$, then $u$'s distance will not improve and thus $u$ will not contribute any cost during that iteration. Hence, fix $W_i^u$ and condition on the event $Y_i^u = 1$. Conditioned on $Y_i^u = 1$, the chosen $v$ is uniform in $W_i^u$ because $v$ is drawn uniformly from $T_i$. The next set $W_{i+1}^u$ is then exactly $\{w \in W_i^u \mid d_G(w,u) < d_G(v,u)\}$ since $v$ is now the closest center to $u$ among all centers selected so far. In other words, $W_{i+1}^u$ is at most the rank of $d_G(v,u)$ in the sorted list $(d_G(w,u))_{w \in W_i^u}$ minus one. Since $v$ is uniformly random in $W_i^u$, the expected size of $W_{i+1}^u$ is at most $\frac{|W_i^u|}{2}$. Applying Lemma $2.1$ of Abboud et al.~\cite{Abboud} yields that $\mathbb{E}\left(\sum_{i = 1}^{n} Y_i^u\right) \leq O(\log n)$. Therefore, the expected total work attributable to successful distance decreases is
    \begin{align*}W' &= O\!\left(\sum_{u \in V} \deg(u)\cdot \log n + \sum_{u \in V}\log n\right) \\ 
    &= O(m \cdot \log n + n \cdot \log n).
    \end{align*}
    
    So far we have only accounted for the work caused by successful distance decreases. Under the distance cap $r$, the same charging argument applies only to the explored vertices and edges together with their corresponding reinsertions and re-relaxations. Therefore, the expected exploration work becomes
    \[O\!\left((n_r + m_r + n_{r,\Delta} + m_{r,\Delta})\log n\right) = O(\ELCDS\log n).\]

    Each Lazy Capped $\Delta$-stepping invocation additionally incurs a
    phase overhead of $\PLCDS$. Since the SSSP computations are
    performed sequentially, this contributes a total work of
    $O(k_r\PLCDS)$. Combining both terms yields
    
    \[
    \mathbb{E}[W]
    =
    O(\ELCDS\log n+k_r\PLCDS),
    \]
    
    which proves the claimed work bound. Since each SSSP computation
    has depth $\DLCDS$, the total depth is
    $O(k_r\DLCDS)$.
    $\qed$
    \end{proof}

 Thus, replacing Lazy Capped Dijkstra by the Lazy Capped $\Delta$-stepping routine changes only the cost of each individual SSSP computation, while the randomized charging argument establishing the $O(\log n)$ expected number of successful distance decreases per vertex remains unchanged. The substitution of Lazy $\Delta$-stepping or Lazy Capped $\Delta$-stepping for the corresponding Dijkstra routine affects the remaining algorithms analogously: Their correctness arguments remain unchanged, and their running time analyses are obtained by replacing the cost of the underlying SSSP routine with the corresponding work and depth bounds established in Section~\ref{sec:sssp}. We therefore omit these analogous running time analyses for the remaining algorithms.

\subsubsection{Abboud MIS Implementation Details}

As stated above, the implementation of this algorithm uses both lazy updates and distance caps for its shortest path routines. One important detail concerning algorithms employing a binary search to find a suitable cluster radius is how to realize that binary search. While the naive approach would be to start at $r_0 = nW, r_1 = \frac{nW}{2}, r_2 = \frac{nW}{4}, ...$ until the first non-feasible $r_i$ has been found, in practice this is often a vast overestimation, leading to many SSSP calls exploring unnecessarily large portions of the input graph before a suitable $r_i$ can be found. Hence, our implementation first performs a doubling search over candidate radii starting at $r_0 = 2$ until a feasible $r_i$ has been found. While a small $r_i$ might still be an underestimation, the resulting SSSP calls are capped at this small radius and thus terminate early, leading to finding the first feasible $r_i$ faster in practice. After that, we continue as expected with a binary search between $r_i$ and $r_{i-1}$ until the smallest feasible radius has been found. 
\subsection{Omitted Proofs of Section~\ref{sec:thorup}}
\label{sec:thorupProofs}

\begin{proof}[Proof of Lemma~\ref{lm:Thorup main lemma}]
We show that after $ c \log n $ iterations for which $ \delta $ is in the $ (1+\epsilon) $-estimation regime (i.e., $ \delta^* < \delta \leq (1 + \epsilon) \delta^* $ holds) all heavy vertices have been eliminated w.h.p., which in turn means $\delta^*$ has decreased by a constant factor.
Consider the start of the first round and let $W$ denote the remaining vertex set and let
\[
H := \{v\in W : |\Gamma_{\leq r}(v)\cap W| \geq 0.6\delta^*\}
\]
be the heavy vertices at the start of first round.

For round $i$ let $H_{i-1}\subseteq H$ denote the heavy vertices that are still heavy at the start of the round. Let
\[
\gamma_i = \frac{|H_i|}{|H_{i-1}|}
\]
be the factor by which the heavy vertices shrink during round $i$. 
Round $i$ is called \emph{good} if $\gamma_i \leq \gamma$, where $\gamma := \alpha(1-p_0) < 1$, $p_0 \in (0,1)$ is the constant from Thorup's analysis~\cite{Thorup}, and $\alpha > 1$ is chosen accordingly. Since $\delta$ remains within a constant factor of $\delta^*$, by Thorup's analysis, each vertex of $H_{i-1}$ is eliminated with probability at least $p_0$. Hence, linearity of expectation gives
\[
\mathbb{E}[\gamma_i] \leq (1-p_0).
\]
Then, by Markov's inequality,
\begin{align*}
\Pr\!\left[\gamma_i > \gamma\right] \leq \frac{\mathbb{E}[\gamma_i]}{\gamma} \leq \frac{1}{\alpha}.
\end{align*}
Hence
\[
\Pr\!\left[\text{round}\: i \:\text{is good}\right] \geq 1 - \frac{1}{\alpha} := q > 0.
\]
Therefore each round is good with probability at least $q$ and this upper bound can be derived independent of the outcome of previous rounds. Let $Z_i$ be the indicator variable that round $i$ is good and define
\[
Z=\sum_{i=1}^{T} Z_i.
\]
Then $\mathbb{E}[Z] \geq qT$.

After $k$ good rounds the number of surviving vertices of $H$ is at most $|H|\gamma^k$. Since $|H|\le n$, all heavy vertices are eliminated once
\[
\gamma^k n < 1,
\]
which holds for
\[
k_0=\left\lceil\log_{1/\gamma} n\right\rceil.
\]
We now choose $T$ large enough that
\[
qT \ge
\left(
8(c+2)+
\frac{2}{\ln(1/\gamma)}
\right)\ln n.
\]
Then $\mu=\mathbb{E}[Z]\ge qT$ satisfies
\[
\mu\ge
\max\{2k_0,\;8(c+2)\ln n\},
\]
and standard Chernoff bound arguments yield
\[
\Pr[Z<k_0]
\le
\frac1{n^{c+2}}.
\]

Hence, after $O(c\log n)$ rounds, all heavy vertices of the current phase have been eliminated with probability at least $1-\frac1{n^{c+2}}$. Consequently, no radius-$r$ ball contains at least $0.6\delta^*$ remaining vertices, implying that $\delta^*$ decreases by a constant factor.
\end{proof}

\begin{proof}[Proof of Lemma~\ref{lm:thorup}]
    We prove the statement inductively over the phases. In the first phase, the algorithm sets $\delta > n \geq \delta^*$, meaning the invariant always holds at the beginning of the first phase.
    
    Now suppose that the invariant holds at the beginning of phase $i$, that is $\delta_i > \delta_i^*$. At the end of the phase, the algorithm sets $\delta_{i+1}=\tfrac{\delta_i}{1+\epsilon}$. We distinguish two cases:\\
    \textbf{Case 1:}\\
    $\tfrac{\delta_i}{1+\epsilon} > \delta_i^* $. As $ \delta^*$ is non-increasing (i.e., $ \delta_i^* \geq \delta_{i+1}^* $), the invariant hold for phase $i+1$.\\
    \textbf{Case 2:}\\
    $\tfrac{\delta_i}{1+\epsilon} \leq \delta_i^*$. Together with the induction hypothesis, this implies $\delta_i^* < \delta_i \leq (1+\epsilon) \delta_i^*$. 
    Note that the value of $ \delta $ is not changed during the $ c \log n $ iterations of the current phase and $ \delta^* $ is non-increasing.
    Thus, the only possibility for the relationship $ \delta^* < \delta \leq (1 + \epsilon) \delta^* $ ceasing to hold is a decrease of $ \delta^* $ to a value below $ \tfrac{\delta_i}{1 + \epsilon} $.
    We can thus apply Lemma~\ref{lm:Thorup main lemma}, by which such a decrease must happen within $ c \log n $ iterations with high probability.
    Thus, $ \delta_{i+1}^* < \tfrac{\delta_i}{1 + \epsilon} = \delta_{i+1} $ holds as desired with high probability.
    
    We conclude that, conditioned on the invariant holding at the beginning of phase $i$, it also holds at the beginning of phase $i+1$ with probability at least $1-\frac{1}{n^{c+2}}$. Since the algorithm performs only $O(\log n)$ phases, a union bound implies that the invariant holds at the beginning of every phase and after every update of $\delta$ with high probability.
\end{proof}

\begin{proof}[Proof of Theorem~\ref{thm:thorup}]
Consider the final of the $ O (\log_{1 + \epsilon} n) $ phases, for which $ \delta = 1 $.
In that phase, the algorithm sets $ R = W $. By the invariant of Lemma~\ref{lm:thorup} we have $ \delta^* < \delta = 1 $ and thus $ \delta^* = 0 $ with high probability. Consequently, the algorithm with high probability sets $ T = W $ in its final iteration and terminates directly after because all nodes are removed from $ W $. As every phase consists of $ O (\log n) $ rounds, the Theorem follows.
\end{proof}

\begin{proof}
Algorithm~\ref{alg:thorup} performs at most $O(\log_{1+\epsilon} n)$ phases, each consisting of $O(\log n)$ rounds. Thus, the total number of rounds is $O(\log^2 n)$. By Lemma~\ref{lm:thorup}, the invariant $\delta > \delta^*$ holds at the beginning of each phase and immediately after each update of $\delta$ with high probability. Since every vertex belongs to its own distance-$r$ neighborhood, we have $\delta^* \geq 1$ whenever $W \neq \emptyset$. Consequently, with high probability, the algorithm cannot reach $\delta \leq 1$ while $W$ is nonempty. Hence, it terminates without invoking the fallback and the Theorem follows.
\end{proof}

\subsubsection{Simplified Thorup Implementation Details}

Just like the algorithm of Abboud et al., the SSSP computations called by this algorithm use lazy updates and the currently guessed $r$ as a distance cap. As in the implementation of Abboud et al.'s algorithm, all weighted shortest path computations use Lazy Capped $\Delta$-stepping. Therefore, the same work and depth bounds from the previous section apply to each shortest path computation. Since our algorithm only succeeds w.h.p., we simply employ the previously described algorithm of Abboud et al. as a fallback which gets invoked if $\delta = 1$ and $W \neq \emptyset$, i.e. we did not perform enough rounds while $\delta$ was in the correct regime \footnote{In theory, a polylogarithmic depth bound can be obtained by instead employing the following fallback: First, compute the all-pairs distances using a repeated squaring approach, then use a parallel MIS algorithm inside of the Hochbaum-Shmoys algorithm. This approach has polynomial work bound but is only use with small enough probability}. Abboud's algorithm then uses the already computed independent set as a starting point. Since Abboud's algorithm also employs lazy updates, the distances computed by Thorup's algorithm can be reused directly. The step A.2.3 in Thorup's algorithm~\cite{Thorup} was realized through multi-source variants of our SSSP routines capped at distance $r$ to identify the vertices that form the set $T$. Step A.2.5 (removing all vertices from $W$ at distance $\leq r$ to any vertex in $U$) was implemented analogously.
\subsection{Gonzalez Implementation Details}
\label{sec:GonzalezImp}

While the $2$-approximation of Gonzalez holds for an arbitrary choice of the first center, we select it uniformly at random to also avoid being biased by the sequence in which the input of a dataset is given. On weighted graphs, the implementation uses Lazy Dijkstra in the sequential setting and Lazy $\Delta$-stepping in the parallel setting. Since Gonzalez' algorithm is not radius based, we do not employ a distance cap here. The implementation is therefore straightforward.
\subsection{Augmented SSSP and Approximate Gonzalez Implementation Details}
\label{sec:ApproximateGonzalezImp}
\begin{algorithm}
\caption{\textsc{Augm.Dijkstra}$(G, u, d, \alpha, H_1,\ldots,H_B)$}
\label{alg:augmented-dijkstra}
\begin{algorithmic}[1]

\Require Graph $G=(V,E,w)$, starting vertex $u$, distance array $d$, bags $H_1,\ldots,H_B$
\Ensure Updated distances $d$ and bags $H_1,\ldots,H_B$

\State $d[u] \gets 0$
\State $Q \gets$ priority queue containing $(u,0)$

\While{$Q \neq \emptyset$}
    \State $(x,d[x]) \gets Q.\textproc{FindMin}()$
    \State $Q.\textproc{DeleteMin}()$

    \ForAll{$(x,y) \in E$}
        \If{$d[y] > d[x] + w(x,y)$}
            \State $H_{\lceil \log_{\alpha} d[y] \rceil} \gets H_{\lceil \log_{\alpha} d[y] \rceil} \setminus \{y\}$
            \State $d[y] \gets d[x] + w(x,y)$
            \State $H_{\lceil \log_{\alpha} d[y] \rceil} \gets H_{\lceil \log_{\alpha} d[y] \rceil} \cup \{y\}$
            \State $Q.\textproc{DecreaseKey}(y,d[y])$
        \EndIf
    \EndFor
\EndWhile

\end{algorithmic}
\end{algorithm}

Since all the augmented SSSP routines have to do in addition to the distance computation is the book keeping of the bags, they just use a flag to switch bag maintenance on/off. The parameter $\epsilon$ which determines $\alpha$ can be passed manually. The rest of the implementation is straightforward. 
\subsection{Algorithmic Parameter Choices}
\label{sec:parameters}

As mentioned, we performed tests on a few instances before the actual experiments to determine algorithm-specific parameter choices. The goal was to fix these parameters for the full experimental sweep over all instances. Since the algorithms of Gonzalez as well as the maximal distance-$r$ independent set algorithm by Abboud et al. do not have such parameters, they are not listed here. Since these tests were done prior to the full experimental sweep across all instances, the timeouts chosen may differ for each graph class. We now discuss the results of these preliminary tests.

\subsubsection{$\Delta$-stepping}

This SSSP routine has a parameter $\Delta$ which determines the width of the distance buckets, i.e. the $i$-th bucket contains all vertices $v$ satisfying $i\Delta \leq d_v < (i+1)\Delta$. We tested different values for a variety of smaller and larger weighted graphs, both sparse and dense. 

\begin{figure}[H]
\centering
\begin{minipage}{0.88\linewidth}
\centering
\captionsetup{width=\linewidth}

\includegraphics[width=0.9\linewidth]{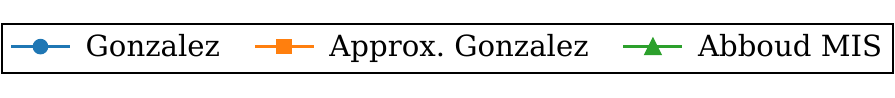}\\[-0.1em]

\begin{subfigure}[t]{0.48\linewidth}
\centering
\includegraphics[width=\linewidth]{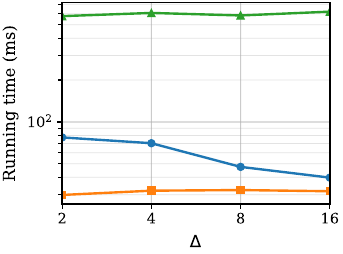}
\caption{$n=10^5$, $\rho=4$, $k=20$}
\end{subfigure}
\hfill
\begin{subfigure}[t]{0.48\linewidth}
\centering
\includegraphics[width=\linewidth]{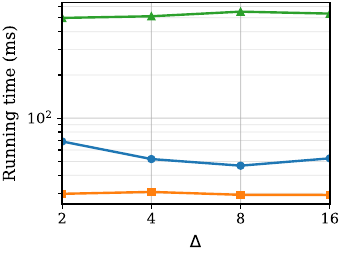}
\caption{$n=10^5$, $\rho=8$, $k=20$}
\end{subfigure}

\par\vspace{-0.15em}

\begin{subfigure}[t]{0.48\linewidth}
\centering
\includegraphics[width=\linewidth]{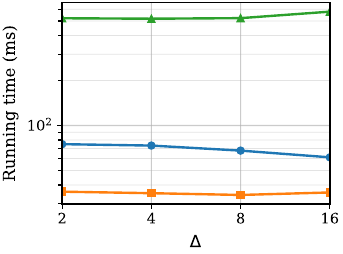}
\caption{$n=10^5$, $\rho=16$, $k=20$}
\end{subfigure}
\hfill
\begin{subfigure}[t]{0.48\linewidth}
\centering
\includegraphics[width=\linewidth]{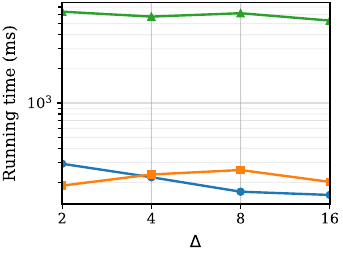}
\caption{$n=10^5$, $\rho=4$, $k=200$}
\end{subfigure}

\par\vspace{-0.15em}

\begin{subfigure}[t]{0.48\linewidth}
\centering
\includegraphics[width=\linewidth]{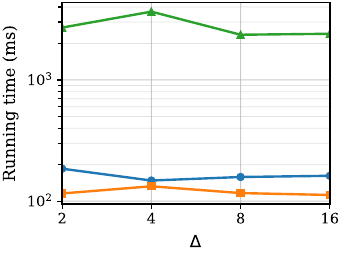}
\caption{$n=10^5$, $\rho=8$, $k=200$}
\end{subfigure}
\hfill
\begin{subfigure}[t]{0.48\linewidth}
\centering
\includegraphics[width=\linewidth]{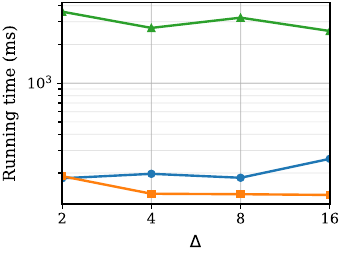}
\caption{$n=10^5$, $\rho=16$, $k=200$}
\end{subfigure}

\par\vspace{-0.15em}

\begin{subfigure}[t]{0.48\linewidth}
\centering
\includegraphics[width=\linewidth]{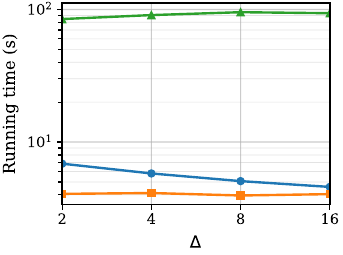}
\caption{$n=10^7$, $\rho=4$, $k=200$}
\end{subfigure}
\hfill
\begin{subfigure}[t]{0.48\linewidth}
\centering
\includegraphics[width=\linewidth]{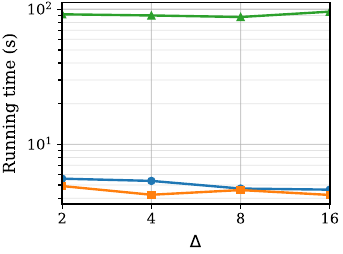}
\caption{$n=10^7$, $\rho=8$, $k=200$}
\end{subfigure}

\par\vspace{-0.15em}

\begin{subfigure}[t]{0.48\linewidth}
\centering
\includegraphics[width=\linewidth]{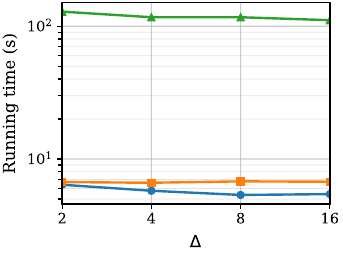}
\caption{$n=10^7$, $\rho=16$, $k=200$}
\end{subfigure}
\hfill
\begin{subfigure}[t]{0.48\linewidth}
\centering
\includegraphics[width=\linewidth]{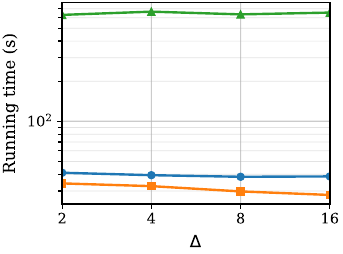}
\caption{$n=10^7$, $\rho=4$, $k=2000$}
\end{subfigure}

\par\vspace{-0.15em}

\begin{subfigure}[t]{0.48\linewidth}
\centering
\includegraphics[width=\linewidth]{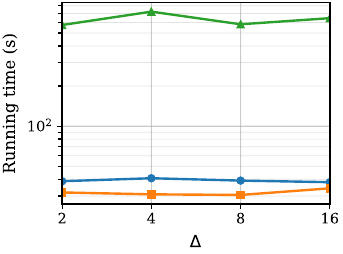}
\caption{$n=10^7$, $\rho=8$, $k=2000$}
\end{subfigure}
\hfill
\begin{subfigure}[t]{0.48\linewidth}
\centering
\includegraphics[width=\linewidth]{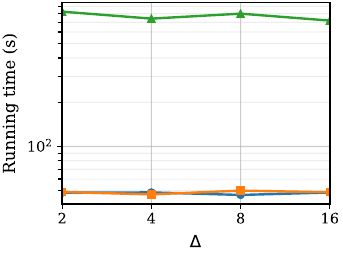}
\caption{$n=10^7$, $\rho=16$, $k=2000$}
\end{subfigure}

\caption{Influence of the bucket width $\Delta$ on the running time of the parallel weighted implementations.}
\label{fig:delta-parameter}
\end{minipage}
\end{figure}

The results indicate that the choice of $\Delta$ only has a minor influence on the running time across all algorithms and instances tested. For most configurations, the running times vary by at most approximately $10\%$, and many of the absolute differences are below one second. The Simplified Thorup algorithm does not appear in these plots, since its execution time was generally higher than its competitors on graphs with high diameter which would therefore make these plots not as visually clear. However, the results was similar in that different values of $\Delta$ did not substantially affect the running time. Hence, we use $\Delta=8$ throughout the remaining experiments.\\

\subsubsection{Approximate Gonzalez}

This algorithm has a parameter $\epsilon$ which determines the size of the buckets used to select the next center. \\
\begin{figure}[H]
\centering
\begin{minipage}{0.8\linewidth}
\centering

\includegraphics[width=0.8\linewidth]{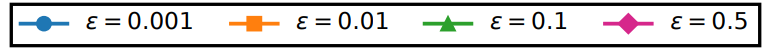}\\[-0.1em]

\begin{subfigure}[t]{\linewidth}
\centering
\includegraphics[width=\linewidth]{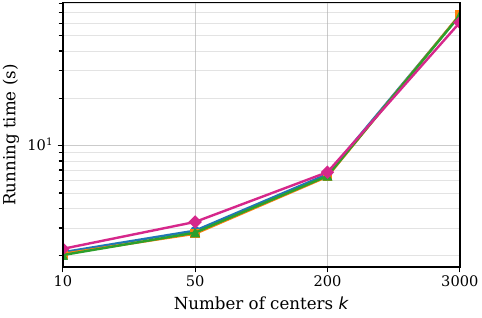}
\caption{Running time}
\end{subfigure}

\par\vspace{-0.15em}

\begin{subfigure}[t]{\linewidth}
\centering
\includegraphics[width=\linewidth]{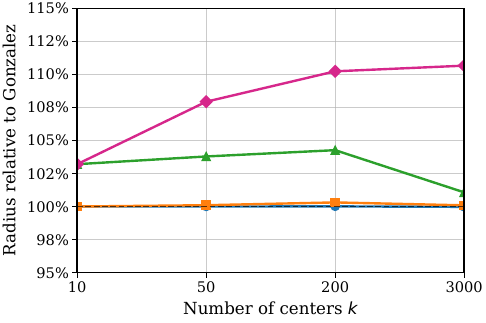}
\caption{Solution quality relative to vanilla Gonzalez}
\end{subfigure}

\caption{Influence of the approximation parameter $\epsilon$ on the
$\alpha$-Approximate Gonzalez algorithm for a weighted synthetic graph
with $n=10^7$ and $\rho=16$. All measurements use parallel
execution with $64$ threads.}
\label{fig:approx-gonzalez-epsilon-synthetic}

\end{minipage}
\end{figure}

\begin{figure}[H]
\centering
\begin{minipage}{0.88\linewidth}
\centering

\includegraphics[width=0.40\linewidth]{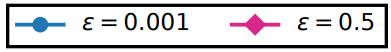}\\[-0.1em]

\begin{subfigure}[t]{0.48\linewidth}
\centering
\includegraphics[width=\linewidth]{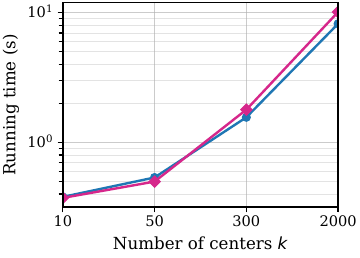}
\caption{LiveJournal: time}
\end{subfigure}
\hfill
\begin{subfigure}[t]{0.48\linewidth}
\centering
\includegraphics[width=\linewidth]{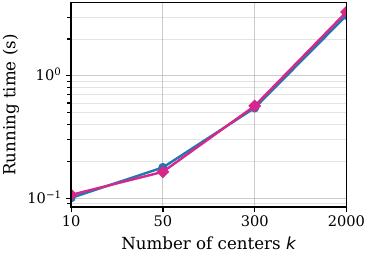}
\caption{YouTube: time}
\end{subfigure}

\par\vspace{-0.15em}

\begin{subfigure}[t]{0.48\linewidth}
\centering
\includegraphics[width=\linewidth]{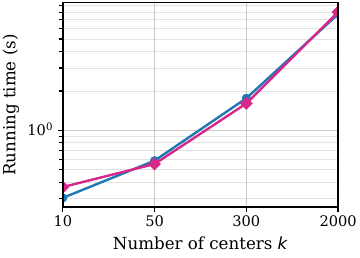}
\caption{Orkut: time}
\end{subfigure}
\hfill
\begin{subfigure}[t]{0.48\linewidth}
\centering
\includegraphics[width=\linewidth]{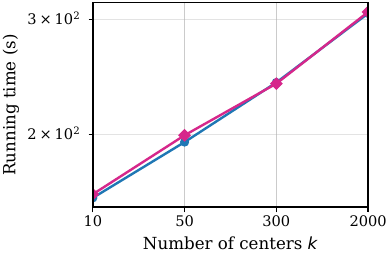}
\caption{USA-central: time}
\end{subfigure}

\par\vspace{-0.15em}

\begin{subfigure}[t]{0.48\linewidth}
\centering
\includegraphics[width=\linewidth]{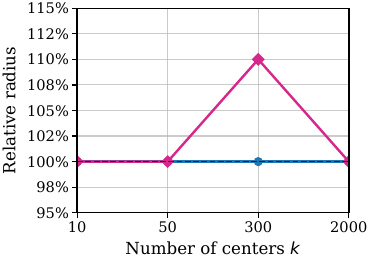}
\caption{LiveJournal: quality}
\end{subfigure}
\hfill
\begin{subfigure}[t]{0.48\linewidth}
\centering
\includegraphics[width=\linewidth]{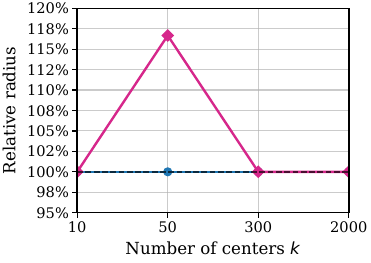}
\caption{YouTube: quality}
\end{subfigure}

\par\vspace{-0.15em}

\begin{subfigure}[t]{0.48\linewidth}
\centering
\includegraphics[width=\linewidth]{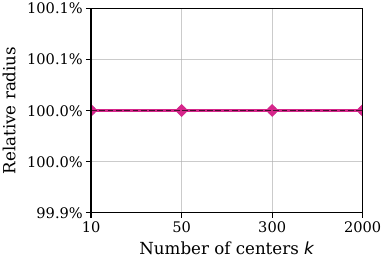}
\caption{Orkut: quality}
\end{subfigure}
\hfill
\begin{subfigure}[t]{0.48\linewidth}
\centering
\includegraphics[width=\linewidth]{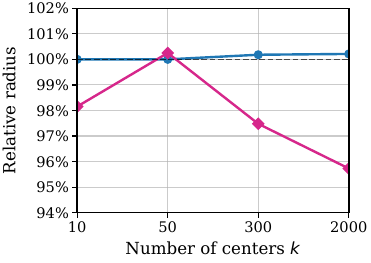}
\caption{USA-central: quality}
\end{subfigure}

\caption{Influence of the approximation parameter $\epsilon$ on the parallel $\alpha$-Approximate Gonzalez algorithm for four real-world graphs.}
\label{fig:approx-gonzalez-epsilon-real-world}

\end{minipage}
\end{figure}

Our tests have shown that a high value of $\epsilon$ such as $\epsilon = 0.5$ does not improve the running time significantly on any of the test instances. However, there is an impact on the solution quality for different $\epsilon$ on all instances. While it seems that on dense graphs a low $\epsilon$ produces a better radius, on sparse weighted graphs with high diameter such as USA-central we observed that a higher value of $\epsilon$ actually improved the solution quality even beyond the Gonzalez radius used as a baseline. We think this is due to the fact that because this graph is very sparse ($\rho < 2$), the increased bag size caused by higher $\epsilon$ values combined with the randomization leads to a better selection of centers than picking the exact farthest vertex. Hence, for our full experimental sweep we chose $\epsilon = 0.5$ for road network graphs and $\epsilon = 0.001$ for all other graphs.\\
\subsubsection{Simplified Thorup's Algorithm}
\label{par:thorup}

As mentioned previously, because our version of Thorup's algorithm does not use Cohen's near-linear-time size estimator, we can choose $\epsilon$ to be a constant.

\begin{figure}[H]
\centering
\begin{minipage}{0.88\linewidth}
\centering
\captionsetup{width=\linewidth}

\includegraphics[width=0.64\linewidth]{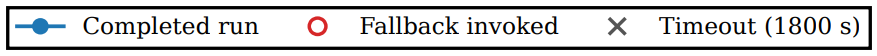}\\[-0.1em]

\begin{subfigure}[t]{0.48\linewidth}
\centering
\includegraphics[width=\linewidth]{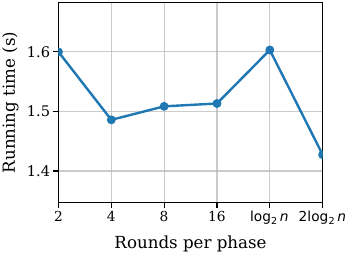}
\caption{$k=20$, $\rho=16$, unw.}
\end{subfigure}
\hfill
\begin{subfigure}[t]{0.48\linewidth}
\centering
\includegraphics[width=\linewidth]{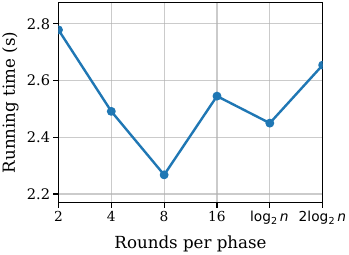}
\caption{$k=100$, $\rho=16$, unw.}
\end{subfigure}

\par\vspace{-0.15em}

\begin{subfigure}[t]{0.48\linewidth}
\centering
\includegraphics[width=\linewidth]{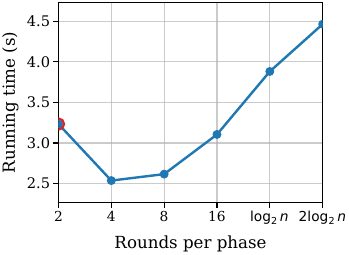}
\caption{$k=500$, $\rho=16$, unw.}
\end{subfigure}
\hfill
\begin{subfigure}[t]{0.48\linewidth}
\centering
\includegraphics[width=\linewidth]{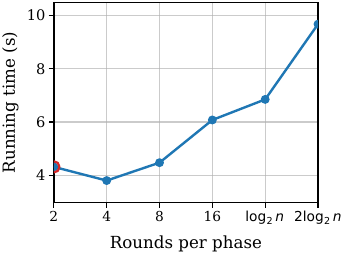}
\caption{$k=3000$, $\rho=16$, unw.}
\end{subfigure}

\par\vspace{-0.15em}

\begin{subfigure}[t]{0.48\linewidth}
\centering
\includegraphics[width=\linewidth]{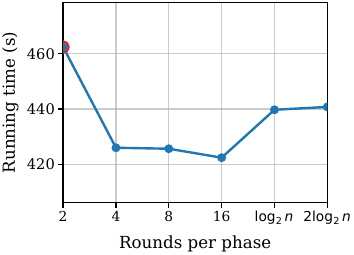}
\caption{$k=20$, $\rho=4$, wgh.}
\end{subfigure}
\hfill
\begin{subfigure}[t]{0.48\linewidth}
\centering
\includegraphics[width=\linewidth]{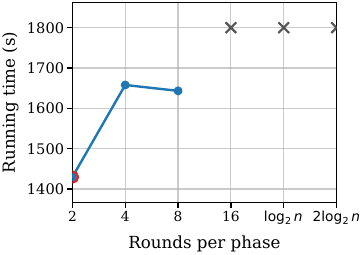}
\caption{$k=100$, $\rho=4$, wgh.}
\end{subfigure}

\par\vspace{-0.15em}

\begin{subfigure}[t]{0.48\linewidth}
\centering
\includegraphics[width=\linewidth]{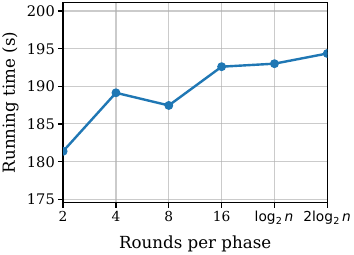}
\caption{$k=20$, $\rho=16$, wgh.}
\end{subfigure}
\hfill
\begin{subfigure}[t]{0.48\linewidth}
\centering
\includegraphics[width=\linewidth]{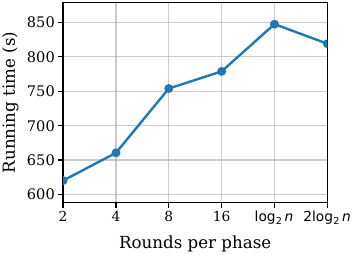}
\caption{$k=100$, $\rho=16$, wgh.}
\end{subfigure}

\caption{Influence of the number of rounds per phase on the running time of the Simplified Thorup algorithm on large synthetic graphs. A red ring marks configurations in which the fallback algorithm was invoked. An $\times$ denotes a run that reached the test timeout. All measurements use $64$ threads and parallel execution.}
\label{fig:thorup-rpp}

\end{minipage}
\end{figure}

The number of rounds per phase (rpp), i.e. per guessed value of $\delta$, should theoretically scale with $\log n$ to obtain the desired high-probability guarantee. In practice, however, the fallback algorithm was triggered only rarely, even with $\epsilon=1$ and a small constant number of rounds per phase. We therefore choose $\epsilon=1$, which halves the guessed value of $\delta$ after each phase and therefore minimizes the number of phases. We set the number of rounds per phase to $4$, which provides the best running time among the tested choices without entering fallbacks. Note that this constant number of rounds per phase is an empirical optimization and does not satisfy the assumptions of the high probability termination guarantee in Theorem~\ref{thm:thorup} which requires $O(\log n)$ rounds per phase.
\subsubsection{Parallelism on Small Synthetic Graphs}
\label{sec:parallelsmall}
\begin{figure}
\centering
\begin{minipage}{0.88\linewidth}
\centering
\includegraphics[width=0.9\linewidth]{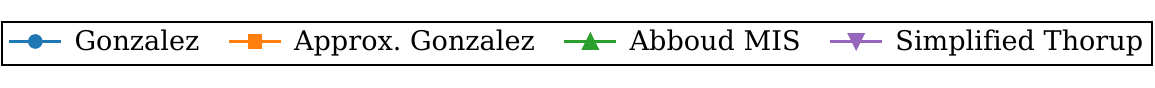}\\

\begin{subfigure}[t]{0.45\linewidth}
\centering
\includegraphics[width=\linewidth]{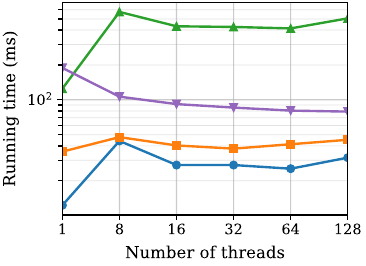}
\caption{$k=20$, $\rho=4$, unw.}
\end{subfigure}
\hfill
\begin{subfigure}[t]{0.45\linewidth}
\centering
\includegraphics[width=\linewidth]{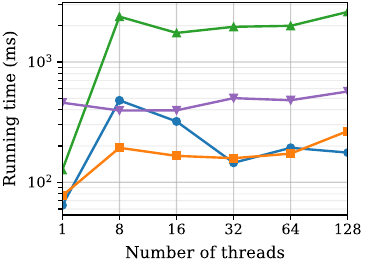}
\caption{$k=200$, $\rho=4$, unw.}
\end{subfigure}
\\[-0.15em]

\begin{subfigure}[t]{0.45\linewidth}
\centering
\includegraphics[width=\linewidth]{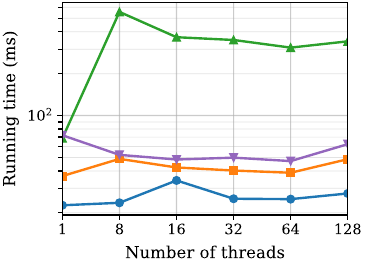}
\caption{$k=20$, $\rho=16$, unw.}
\end{subfigure}
\hfill
\begin{subfigure}[t]{0.45\linewidth}
\centering
\includegraphics[width=\linewidth]{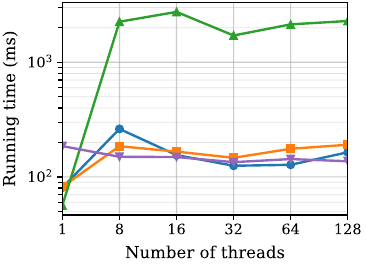}
\caption{$k=200$, $\rho=16$, unw.}
\end{subfigure}
\\[-0.15em]

\begin{subfigure}[t]{0.45\linewidth}
\centering
\includegraphics[width=\linewidth]{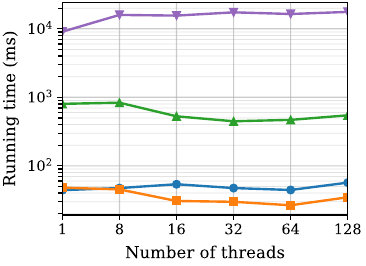}
\caption{$k=20$, $\rho=4$, wgh.}
\end{subfigure}
\hfill
\begin{subfigure}[t]{0.45\linewidth}
\centering
\includegraphics[width=\linewidth]{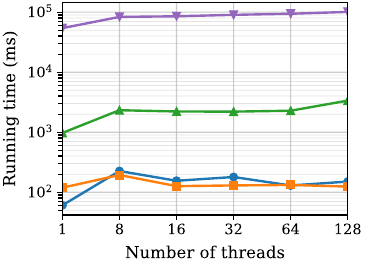}
\caption{$k=200$, $\rho=4$, wgh.}
\end{subfigure}
\\[-0.15em]

\begin{subfigure}[t]{0.45\linewidth}
\centering
\includegraphics[width=\linewidth]{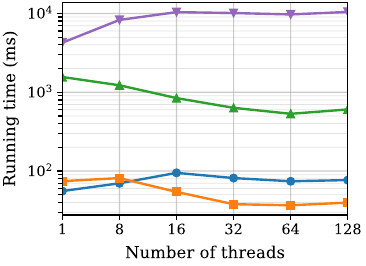}
\caption{$k=20$, $\rho=16$, wgh.}
\end{subfigure}
\hfill
\begin{subfigure}[t]{0.45\linewidth}
\centering
\includegraphics[width=\linewidth]{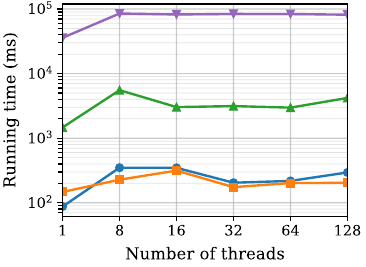}
\caption{$k=200$, $\rho=16$, wgh.}
\end{subfigure}
\\[-0.15em]

\caption{Parallel scaling on small synthetic graphs with $n=10^5$. Each row corresponds to one graph and each column to one value of $k$. The lower left corner of the x-axis represents single-core variants.}
\label{fig:parallel-scaling-small-synthetic}
\end{minipage}
\end{figure}

On the small synthetic graphs, parallelization does not pay off. Compared with their single-core executions, all algorithms either achieve only minor improvements or become substantially slower. For graphs of this size, the overhead of parallel execution dominates the computation, independently of the value of $k$.
\subsubsection{Parallelism on Large Synthetic Graphs}
\label{sec:parallellarge}
\begin{figure}
\centering
\begin{minipage}{0.88\linewidth}
\centering
\includegraphics[width=0.9\linewidth]{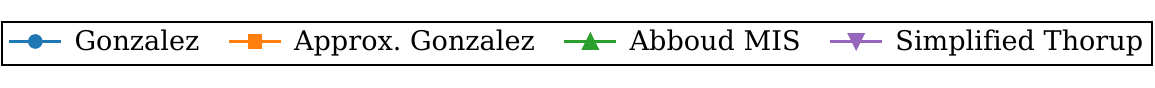}\\

\begin{subfigure}[t]{0.48\linewidth}
\centering
\includegraphics[width=\linewidth]{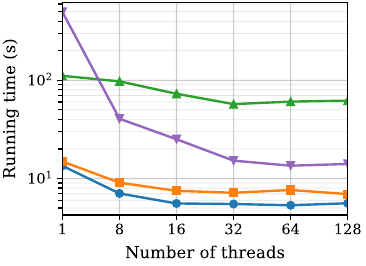}
\caption{$k=200$, $\rho=4$, unw.}
\end{subfigure}
\hfill
\begin{subfigure}[t]{0.48\linewidth}
\centering
\includegraphics[width=\linewidth]{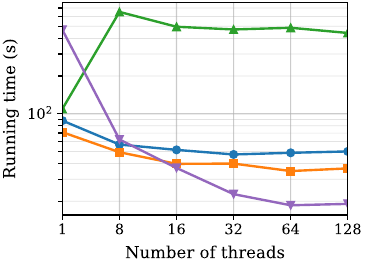}
\caption{$k=2000$, $\rho=4$, unw.}
\end{subfigure}
\\[-0.15em]

\begin{subfigure}[t]{0.48\linewidth}
\centering
\includegraphics[width=\linewidth]{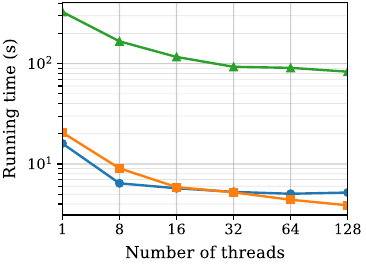}
\caption{$k=200$, $\rho=4$, wgh.}
\end{subfigure}
\hfill
\begin{subfigure}[t]{0.48\linewidth}
\centering
\includegraphics[width=\linewidth]{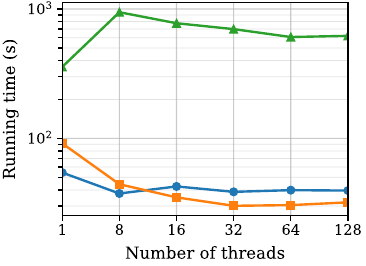}
\caption{$k=2000$, $\rho=4$, wgh.}
\end{subfigure}
\\[-0.15em]

\begin{subfigure}[t]{0.48\linewidth}
\centering
\includegraphics[width=\linewidth]{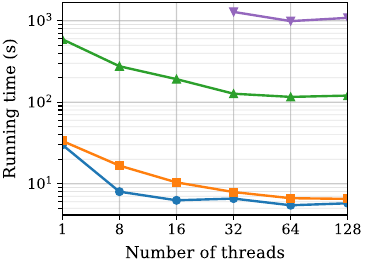}
\caption{$k=200$, $\rho=16$, wgh.}
\end{subfigure}
\hfill
\begin{subfigure}[t]{0.48\linewidth}
\centering
\includegraphics[width=\linewidth]{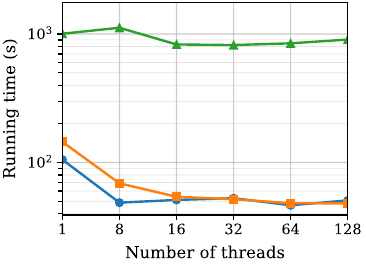}
\caption{$k=2000$, $\rho=16$, wgh.}
\end{subfigure}
\\[-0.15em]

\caption{Parallel scaling on large synthetic graphs with $n=10^7$, using the same layout as Figure~\ref{fig:parallel-scaling-small-synthetic}}.
\label{fig:parallel-scaling-large-synthetic}
\end{minipage}
\end{figure}

On the large synthetic graphs we can observe a different picture: With the exception of Abboud MIS for large values of $k$, all algorithms significantly benefit from parallelism. For the lower value $k=200$, the two Gonzalez algorithms improve their running times by a factor of roughly $2$. Abboud MIS also improves by a factor of $2-5$, depending on the density of the instance. The Simplified Thorup algorithm improves significantly on unweighted graphs. However, on weighted graphs this algorithm struggles a lot, timing out on most runs even on large thread counts and finishing far behind all other algorithms when the runs did finish. We believe that this difference is caused by the interaction between the number of candidate radii, the number of phases performed by the algorithm and the cost of its multi-source SSSP computations. Unweighted instances generally have smaller distance ranges which in turn tends to reduce the number of feasibility checks. Within each feasibility check, the multi-source SSSP computations benefit substantially from parallelism. On weighted graphs, the larger distance range and more expensive capped shortest-path computations cause the additional phases to dominate the running time. All improvements seem to stagnate around $64$ threads. Executions with $128$ threads do not yield significant improvements (if any). This suggests that, beyond this point, synchronization and memory bandwidth overhead outweigh the additional available parallelism.\\

For larger values of $k$ approaching $\sqrt{n}$, the farthest point based algorithms keep their improvements by similar margins. However, Abboud MIS now performs substantially worse on sparse graphs in the parallel execution than the single-core variant. We believe this is due to both the nature of the algorithm as well as our implementation of the binary search: Starting with $r=1$, the algorithm now performs $\sqrt{n}$ many SSSP calls of very low depth, each causing a parallel overhead. Since sparse graphs tend to have a high diameter, many candidate radii are needed, each entailing the same process. On sufficiently dense instances, the explored regions become larger and parallelism begins to pay off again, as shown in Figure~\ref{fig:parallel-scaling-large-synthetic}, subfigure~(f). The simplified Thorup algorithm shows a related trend for larger values of $k$: it benefits significantly from parallelism on unweighted graphs, but the additional feasibility checks and phases on weighted graphs still cause most configurations to time out.
\subsubsection{Small Synthetic Graphs}
\label{sec:ResultsSmall}

\begin{figure}
\centering
\begin{minipage}{0.88\linewidth}
\centering

\includegraphics[width=0.90\linewidth]{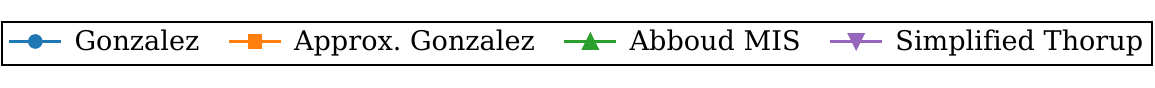}\\[-0.1em]

\begin{subfigure}[t]{0.48\linewidth}
\centering
\includegraphics[width=\linewidth]{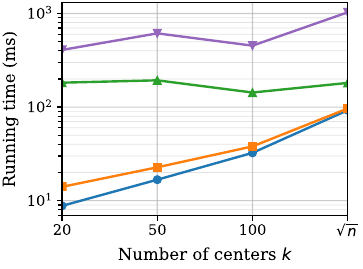}
\caption{$\rho=2$, unw.}
\end{subfigure}
\hfill
\begin{subfigure}[t]{0.48\linewidth}
\centering
\includegraphics[width=\linewidth]{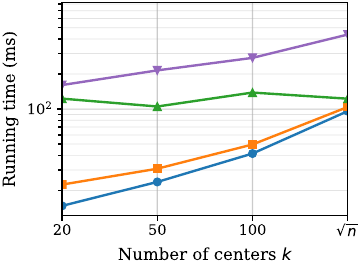}
\caption{$\rho=4$, unw.}
\end{subfigure}

\par\vspace{-0.15em}

\begin{subfigure}[t]{0.48\linewidth}
\centering
\includegraphics[width=\linewidth]{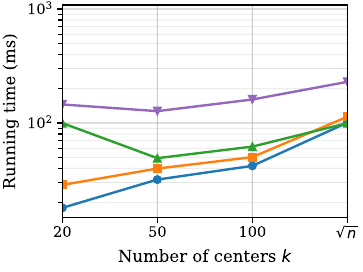}
\caption{$\rho=8$, unw.}
\end{subfigure}
\hfill
\begin{subfigure}[t]{0.48\linewidth}
\centering
\includegraphics[width=\linewidth]{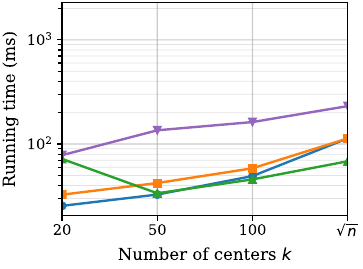}
\caption{$\rho=16$, unw.}
\end{subfigure}

\par\vspace{-0.15em}

\begin{subfigure}[t]{0.48\linewidth}
\centering
\includegraphics[width=\linewidth]{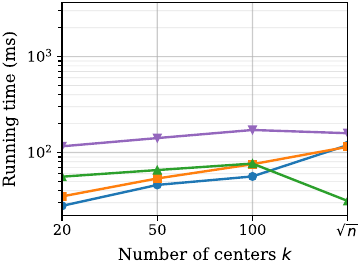}
\caption{$\rho=32$, unw.}
\end{subfigure}
\hfill
\begin{subfigure}[t]{0.48\linewidth}
\centering
\includegraphics[width=\linewidth]{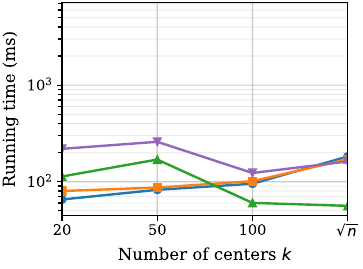}
\caption{$\rho=64$, unw.}
\end{subfigure}

\par\vspace{-0.15em}

\begin{subfigure}[t]{0.48\linewidth}
\centering
\includegraphics[width=\linewidth]{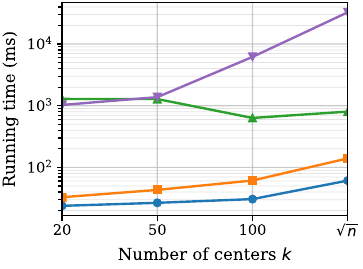}
\caption{$\rho=2$, wgh.}
\end{subfigure}
\hfill
\begin{subfigure}[t]{0.48\linewidth}
\centering
\includegraphics[width=\linewidth]{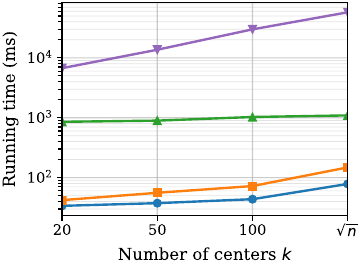}
\caption{$\rho=4$, wgh.}
\end{subfigure}

\par\vspace{-0.15em}

\begin{subfigure}[t]{0.48\linewidth}
\centering
\includegraphics[width=\linewidth]{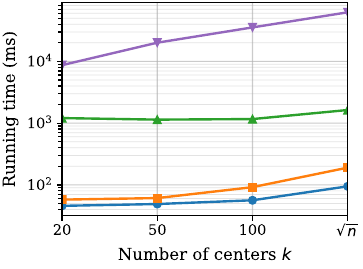}
\caption{$\rho=8$, wgh.}
\end{subfigure}
\hfill
\begin{subfigure}[t]{0.48\linewidth}
\centering
\includegraphics[width=\linewidth]{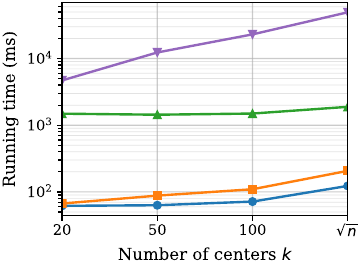}
\caption{$\rho=16$, wgh.}
\end{subfigure}

\par\vspace{-0.15em}

\begin{subfigure}[t]{0.48\linewidth}
\centering
\includegraphics[width=\linewidth]{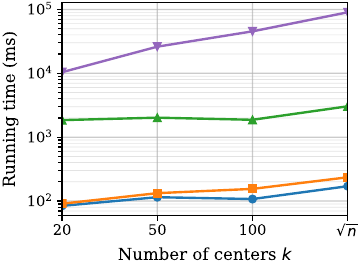}
\caption{$\rho=32$, wgh.}
\end{subfigure}
\hfill
\begin{subfigure}[t]{0.48\linewidth}
\centering
\includegraphics[width=\linewidth]{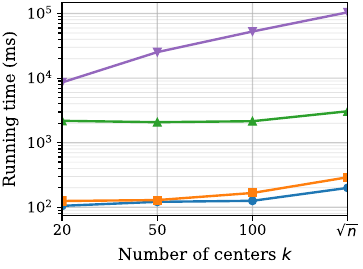}
\caption{$\rho=64$, wgh.}
\end{subfigure}

\caption{Running times on small synthetic graphs of size $n=10^5$. All algorithms were executed in single-core mode only.}
\label{fig:small-synthetic-full-sweep-runtime}

\end{minipage}
\end{figure}

We first analyze the running times on the small generated graphs with $n=10^5$. We consider increasing values of $k$ up to approximately $\sqrt{n}$. Since the preceding experiments showed that parallelism does not pay off at this graph size, all results in Figure~\ref{fig:small-synthetic-full-sweep-runtime} use single-core execution.\\

All algorithms finish quickly on the unweighted instances, with the longest executions taking approximately one second for Simplified Thorup on sparse graphs. As the density increases, the diameter decreases, reducing the number of feasibility checks required by the MIS-based algorithms. Consequently, their running times decrease. For large $k$, Abboud MIS even outperforms both Gonzalez variants, as shown in Figure~\ref{fig:small-synthetic-full-sweep-runtime}, subfigures~(d)--(f).\\

On the weighted instances, the MIS-based algorithms perform worse than the farthest-point-based approaches. This is not surprising, since the larger distance range causes the search for the smallest feasible radius to require more candidate radii. However, the running times of the MIS-based algorithms remain relatively stable as the density increases. This can be explained by two opposing effects. On the one hand, higher density generally reduces the graph diameter and therefore the number of feasibility checks. On the other hand, it increases the work performed by each SSSP computation. For graphs of this size, these two effects appear to approximately balance each other out.

\begin{figure}
\centering
\begin{minipage}{0.88\linewidth}
\centering

\includegraphics[width=0.8\linewidth]{Plots/small_full_sweep_runtime_legend.pdf}\\

\begin{subfigure}[t]{0.48\linewidth}
\centering
\includegraphics[width=\linewidth]{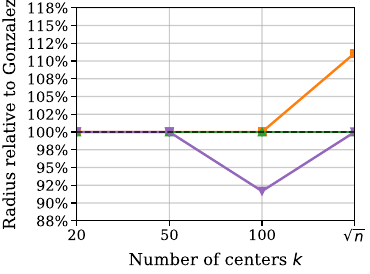}
\caption{$\rho=2$, unw.}
\end{subfigure}
\hfill
\begin{subfigure}[t]{0.48\linewidth}
\centering
\includegraphics[width=\linewidth]{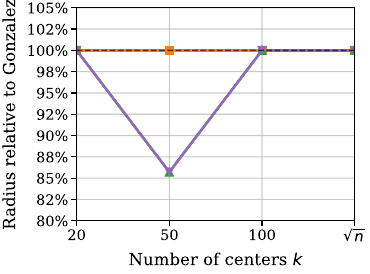}
\caption{$\rho=4$, unw.}
\end{subfigure}

\par\vspace{-0.15em}

\begin{subfigure}[t]{0.48\linewidth}
\centering
\includegraphics[width=\linewidth]{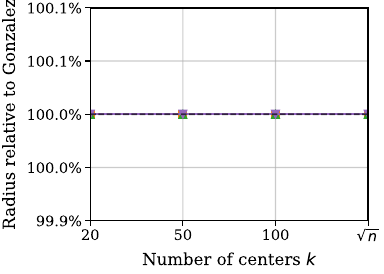}
\caption{$\rho=8$, unw.}
\end{subfigure}
\hfill
\begin{subfigure}[t]{0.48\linewidth}
\centering
\includegraphics[width=\linewidth]{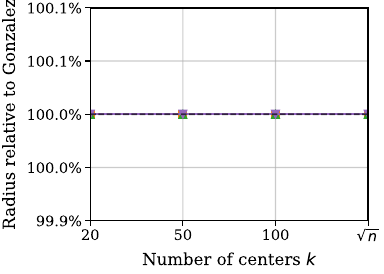}
\caption{$\rho=16$, unw.}
\end{subfigure}

\par\vspace{-0.15em}

\begin{subfigure}[t]{0.48\linewidth}
\centering
\includegraphics[width=\linewidth]{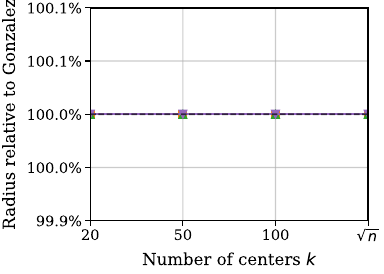}
\caption{$\rho=32$, unw.}
\end{subfigure}
\hfill
\begin{subfigure}[t]{0.48\linewidth}
\centering
\includegraphics[width=\linewidth]{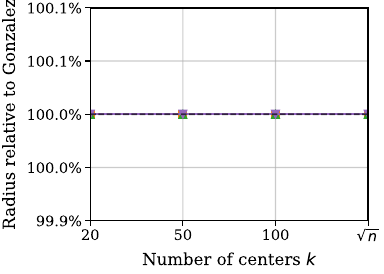}
\caption{$\rho=64$, unw.}
\end{subfigure}

\par\vspace{-0.15em}

\begin{subfigure}[t]{0.48\linewidth}
\centering
\includegraphics[width=\linewidth]{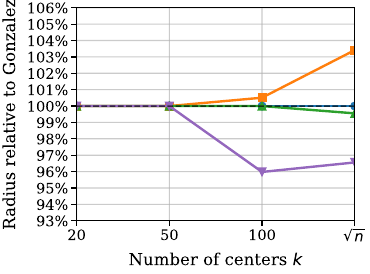}
\caption{$\rho=2$, wgh.}
\end{subfigure}
\hfill
\begin{subfigure}[t]{0.48\linewidth}
\centering
\includegraphics[width=\linewidth]{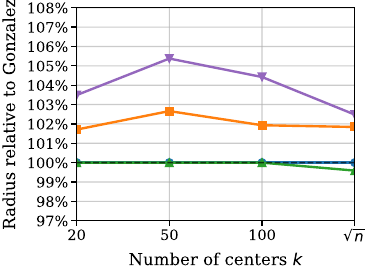}
\caption{$\rho=4$, wgh.}
\end{subfigure}

\par\vspace{-0.15em}

\begin{subfigure}[t]{0.48\linewidth}
\centering
\includegraphics[width=\linewidth]{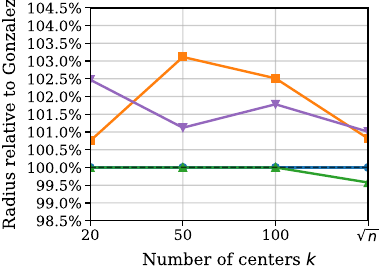}
\caption{$\rho=8$, wgh.}
\end{subfigure}
\hfill
\begin{subfigure}[t]{0.48\linewidth}
\centering
\includegraphics[width=\linewidth]{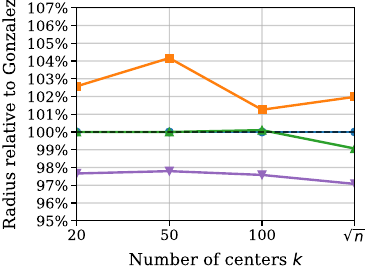}
\caption{$\rho=16$, wgh.}
\end{subfigure}

\par\vspace{-0.15em}

\begin{subfigure}[t]{0.48\linewidth}
\centering
\includegraphics[width=\linewidth]{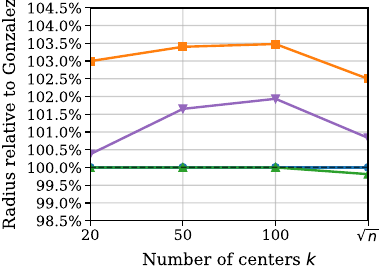}
\caption{$\rho=32$, wgh.}
\end{subfigure}
\hfill
\begin{subfigure}[t]{0.48\linewidth}
\centering
\includegraphics[width=\linewidth]{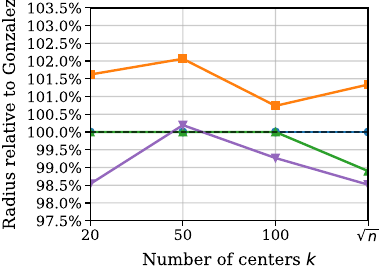}
\caption{$\rho=64$, wgh.}
\end{subfigure}

\caption{Solution quality on small synthetic graphs of size $n=10^5$, measured as the radius returned by each algorithm relative to the radius returned by Gonzalez for the same graph and value of $k$.}
\label{fig:small-synthetic-full-sweep-quality}

\end{minipage}
\end{figure}

For the solution-quality plots, Gonzalez defines the baseline at $100\%$. The $k$-center radius returned by each other algorithm is reported relative to that of Gonzalez for the same graph and value of $k$.\\

On the unweighted instances, the radii returned by the algorithms differ by an absolute value of at most $1$, while the reported radii themselves are typically between $5$ and $7$. The MIS-based algorithms produce a smaller radius than Gonzalez in two configurations, whereas Approximate Gonzalez produces a radius larger by $1$ in one configuration.\\

On the weighted graphs, the same tendency is more visible because the $k$-center radii reported range from approximately $1000$ on the densest instances to around $50000$ on the sparsest instances. Approximate Gonzalez typically returns solutions that are $2$--$5\%$ worse than the Gonzalez baseline. The MIS-based algorithms, in contrast, occasionally improve upon Gonzalez by approximately $2\%$. Recall that Gonzalez is itself only a $2$-approximation algorithm. It is plausible that the randomized MIS-based algorithms therefore occasionally find a better center set.
\subsubsection{Large Synthetic Graphs}
\label{sec:ResultsLarge}

\begin{figure}
\centering
\begin{minipage}{0.88\linewidth}
\centering

\includegraphics[width=0.8\linewidth]{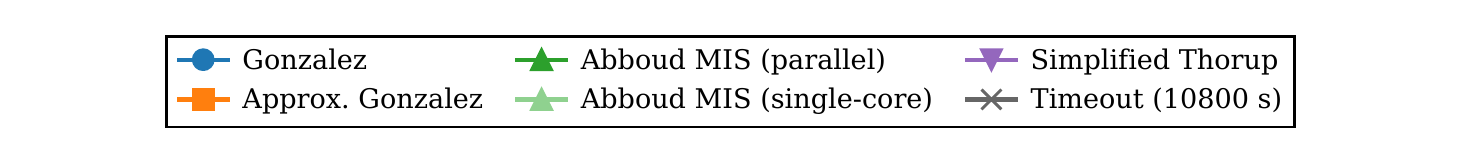}\\[0.1em]

\begin{subfigure}[t]{0.48\linewidth}
\centering
\includegraphics[width=\linewidth]{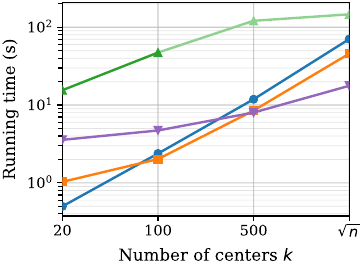}
\caption{$\rho=2$, unw.}
\end{subfigure}
\hfill
\begin{subfigure}[t]{0.48\linewidth}
\centering
\includegraphics[width=\linewidth]{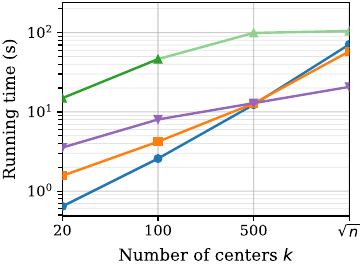}
\caption{$\rho=4$, unw.}
\end{subfigure}

\par\vspace{-0.15em}

\begin{subfigure}[t]{0.48\linewidth}
\centering
\includegraphics[width=\linewidth]{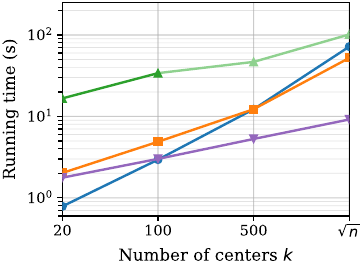}
\caption{$\rho=8$, unw.}
\end{subfigure}
\hfill
\begin{subfigure}[t]{0.48\linewidth}
\centering
\includegraphics[width=\linewidth]{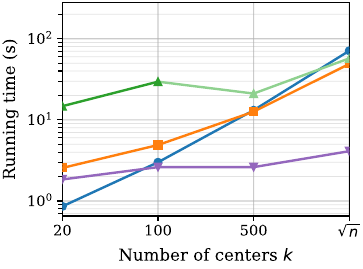}
\caption{$\rho=16$, unw.}
\end{subfigure}

\par\vspace{-0.15em}

\begin{subfigure}[t]{0.48\linewidth}
\centering
\includegraphics[width=\linewidth]{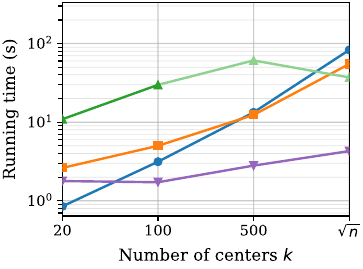}
\caption{$\rho=32$, unw.}
\end{subfigure}
\hfill
\begin{subfigure}[t]{0.48\linewidth}
\centering
\includegraphics[width=\linewidth]{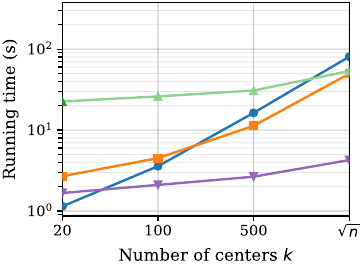}
\caption{$\rho=64$, unw.}
\end{subfigure}

\par\vspace{-0.15em}

\begin{subfigure}[t]{0.48\linewidth}
\centering
\includegraphics[width=\linewidth]{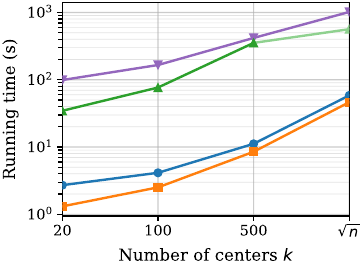}
\caption{$\rho=2$, wgh.}
\end{subfigure}
\hfill
\begin{subfigure}[t]{0.48\linewidth}
\centering
\includegraphics[width=\linewidth]{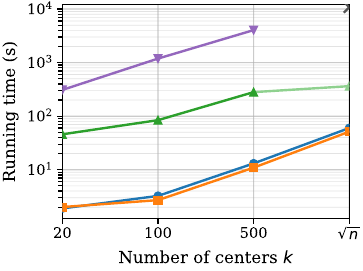}
\caption{$\rho=4$, wgh.}
\end{subfigure}

\par\vspace{-0.15em}

\begin{subfigure}[t]{0.48\linewidth}
\centering
\includegraphics[width=\linewidth]{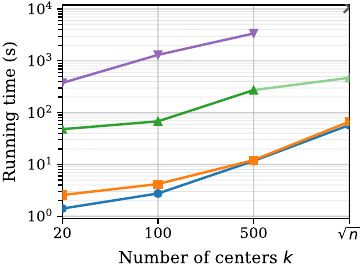}
\caption{$\rho=8$, wgh.}
\end{subfigure}
\hfill
\begin{subfigure}[t]{0.48\linewidth}
\centering
\includegraphics[width=\linewidth]{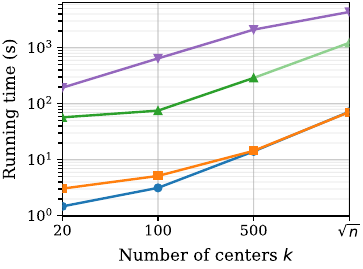}
\caption{$\rho=16$, wgh.}
\end{subfigure}

\par\vspace{-0.15em}

\begin{subfigure}[t]{0.48\linewidth}
\centering
\includegraphics[width=\linewidth]{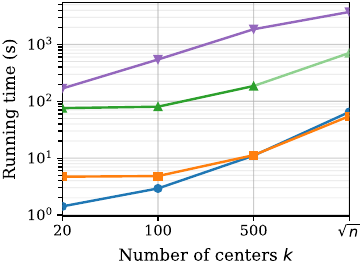}
\caption{$\rho=32$, wgh.}
\end{subfigure}
\hfill
\begin{subfigure}[t]{0.48\linewidth}
\centering
\includegraphics[width=\linewidth]{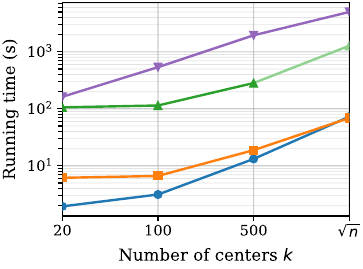}
\caption{$\rho=64$, wgh.}
\end{subfigure}

\caption{Running times on large synthetic graphs with $n=10^7$. Each algorithm uses parallel execution with $64$ threads except Abboud MIS: Each point shows the faster result of the parallel (dark) and single-core (light) runs.}
\label{fig:large-synthetic-full-sweep-runtime}

\end{minipage}
\end{figure}

Our experimental evaluation continues on the large generated instances with $n=10^7$. The preceding scaling experiments showed that parallel execution generally improves the running times on graphs of this size. Abboud MIS is an exception for large $k$, where the overhead of its many sequentially invoked, shallow SSSP computations exploring only small neighborhoods can make the single-core implementation faster. Because of this, for larger graphs, the plots show Abboud MIS' running time performance in parallel mode whenever $k$ is small and in single-core mode once $k$ becomes too large. This is shown in figure \ref{fig:large-synthetic-full-sweep-runtime} by having a dark green dot whenever Abboud MIS is executed in parallel mode and a light green dot when executed in single-core mode.\\

The running-time results are generally consistent with the preceding observations. The farthest-point-based approaches remain relatively close in performance and finish quickly across all instances. The performance of the MIS-based algorithms depends much more on the distance range and diameter of the graph because these determine the number and cost of the feasibility checks. Consequently, the MIS-based algorithms perform particularly well on unweighted graphs, especially at higher densities. For sufficiently large $k$, Simplified Thorup even outperforms all farthest-point-based approaches on the dense unweighted instances, as shown in Figure~\ref{fig:large-synthetic-full-sweep-runtime}, subfigures~(d)--(f).\\

On the weighted instances, the MIS-based algorithms require substantially more feasibility checks and therefore fall behind the farthest-point-based approaches. Among the latter, Gonzalez generally achieves the best running time.

\begin{figure}
\centering
\begin{minipage}{0.88\linewidth}
\centering

\includegraphics[width=0.90\linewidth]{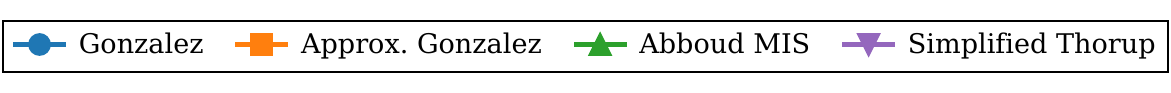}\\[0.2em]

\begin{subfigure}[t]{0.48\linewidth}
\centering
\includegraphics[width=\linewidth]{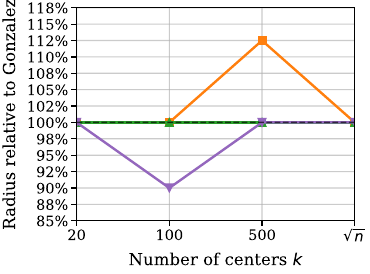}
\caption{$\rho=4$, unw.}
\end{subfigure}
\hfill
\begin{subfigure}[t]{0.48\linewidth}
\centering
\includegraphics[width=\linewidth]{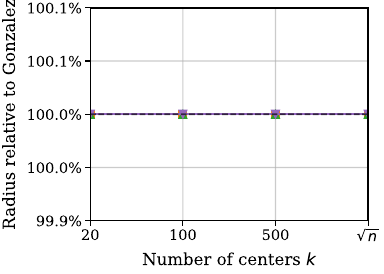}
\caption{$\rho=8$, unw.}
\end{subfigure}

\par\vspace{-0.15em}

\begin{subfigure}[t]{0.48\linewidth}
\centering
\includegraphics[width=\linewidth]{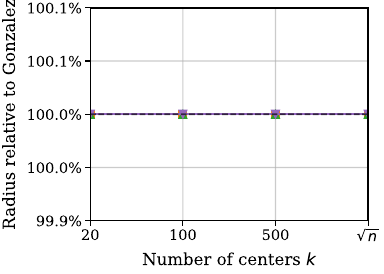}
\caption{$\rho=16$, unw.}
\end{subfigure}
\hfill
\begin{subfigure}[t]{0.48\linewidth}
\centering
\includegraphics[width=\linewidth]{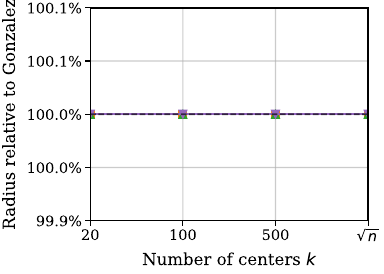}
\caption{$\rho=32$, unw.}
\end{subfigure}

\par\vspace{-0.15em}

\begin{subfigure}[t]{0.48\linewidth}
\centering
\includegraphics[width=\linewidth]{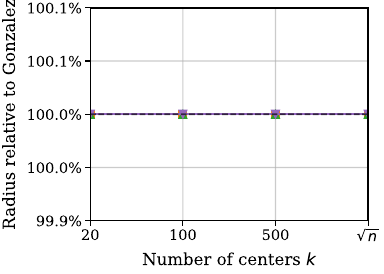}
\caption{$\rho=64$, unw.}
\end{subfigure}
\hfill
\begin{subfigure}[t]{0.48\linewidth}
\centering
\includegraphics[width=\linewidth]{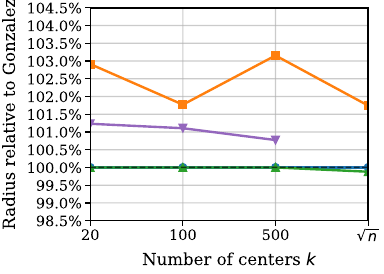}
\caption{$\rho=4$, wgh.}
\end{subfigure}

\par\vspace{-0.15em}

\begin{subfigure}[t]{0.48\linewidth}
\centering
\includegraphics[width=\linewidth]{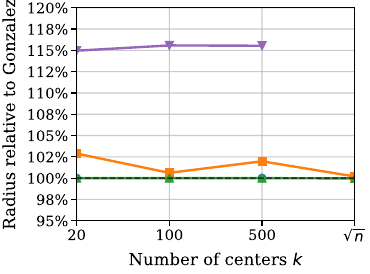}
\caption{$\rho=8$, wgh.}
\end{subfigure}
\hfill
\begin{subfigure}[t]{0.48\linewidth}
\centering
\includegraphics[width=\linewidth]{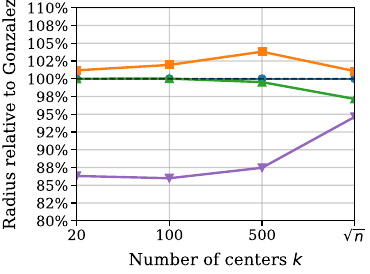}
\caption{$\rho=16$, wgh.}
\end{subfigure}

\par\vspace{-0.15em}

\begin{subfigure}[t]{0.48\linewidth}
\centering
\includegraphics[width=\linewidth]{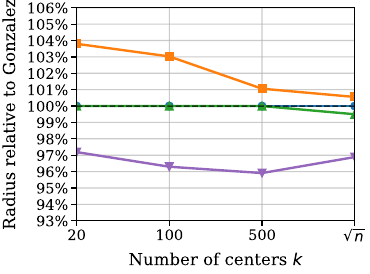}
\caption{$\rho=32$, wgh.}
\end{subfigure}
\hfill
\begin{subfigure}[t]{0.48\linewidth}
\centering
\includegraphics[width=\linewidth]{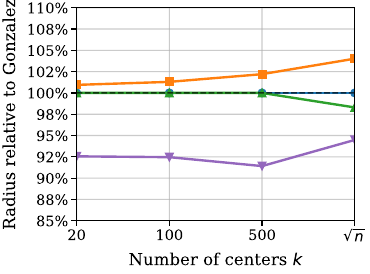}
\caption{$\rho=64$, wgh.}
\end{subfigure}

\caption{Solution quality on large synthetic graphs with $n=10^7$, measured as the radius returned by each algorithm relative to the radius returned by Gonzalez for the same graph and value of $k$. Missing points correspond to configurations that reached the timeout.}
\label{fig:large-synthetic-full-sweep-quality}

\end{minipage}
\end{figure}

The solution quality plots also confirm the trends shown on the small synthetic graphs. The reason the $\rho=2$ instances are omitted in Figure \ref{fig:large-synthetic-full-sweep-quality} is because the graphs generated with this density contained too many connected components and thus all algorithms reported a radius of $\infty$. On unweighted instances of higher density, all algorithms are very close and differ by an absolute value of at most $1$ as the $k$-center radii returned by the algorithms on these instances range approximately from $3$ to $12$.\\

On the weighted instances, Approximate Gonzalez consistently reports worse radii than Gonzalez. Simplified Thorup also performs worse on sparse weighted instances, but its relative solution quality improves as the density increases. On the dense instances, it improves upon the Gonzalez baseline by up to $13\%$ as shown in Figure~\ref{fig:large-synthetic-full-sweep-quality}, subfigure~(h). For sufficiently large $k$, Abboud MIS also reports radii approximately $2\%$ smaller than Gonzalez. Across these instances, the reported $k$-center radii range from approximately $1800$ to $50000$.
\subsection{Social Networks Solution Quality}
\label{sec:socQual}
\begin{figure}
\centering
\begin{minipage}{0.88\linewidth}
\centering

\includegraphics[width=0.9\linewidth]{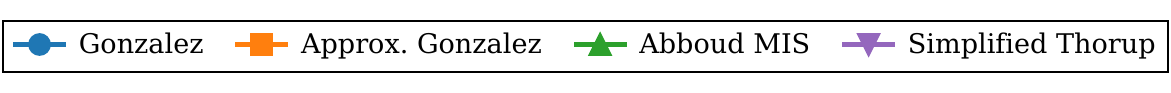}\\[-0.1em]

\begin{subfigure}[t]{0.48\linewidth}
\centering
\includegraphics[width=\linewidth]{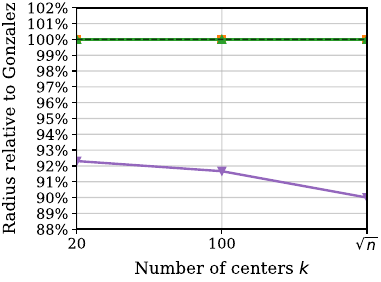}
\caption{DBLP}
\end{subfigure}
\hfill
\begin{subfigure}[t]{0.48\linewidth}
\centering
\includegraphics[width=\linewidth]{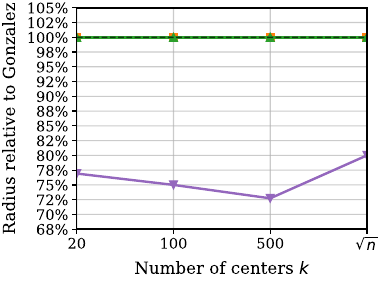}
\caption{YouTube}
\end{subfigure}

\par\vspace{-0.15em}

\begin{subfigure}[t]{0.48\linewidth}
\centering
\includegraphics[width=\linewidth]{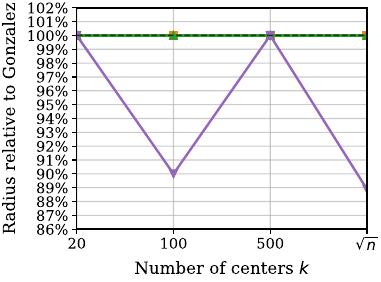}
\caption{LiveJournal}
\end{subfigure}
\hfill
\begin{subfigure}[t]{0.48\linewidth}
\centering
\includegraphics[width=\linewidth]{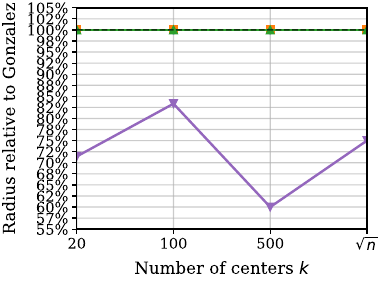}
\caption{Orkut}
\end{subfigure}

\par\vspace{-0.15em}

\begin{subfigure}[t]{0.48\linewidth}
\centering
\includegraphics[width=\linewidth]{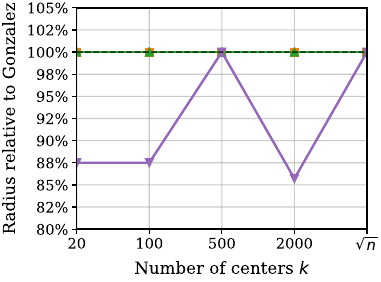}
\caption{Twitter}
\end{subfigure}
\hfill
\begin{subfigure}[t]{0.48\linewidth}
\centering
\includegraphics[width=\linewidth]{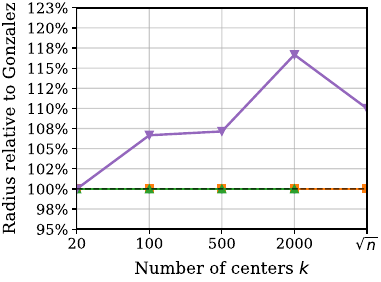}
\caption{Friendster}
\end{subfigure}

\caption{Solution quality on the six social-network instances. The $k$-center radius returned by each algorithm is reported relative to the one reported by Gonzalez for the same graph and value of $k$. The dashed line marks equality with Gonzalez. Missing points correspond to configurations that reached the timeout.}
\label{fig:social-full-sweep-quality}

\end{minipage}
\end{figure}

The quality of solution plots show that for these instances, all algorithms reported identical radii for all values of $k$ except the Simplified Thorup algorithm. Although the relative differences for Simplified Thorup appear large, they should be interpreted in the context of the small absolute radius on these unweighted, low-diameter graphs. Even a difference of two in the reported radius can correspond to a relative improvement of $40\%$. For example, Figure~\ref{fig:social-full-sweep-quality} subfigure~(d) shows that for $k = 500$, Simplified Thorup reports a $k$-center solution of only roughly $60\%$ of the Gonzalez baseline. In absolute terms, however, Simplified Thorup reported a radius of $3$ here vs. the reported Gonzalez radius of $5$.
\subsubsection{Other Networks}
\label{sec:ResultsOther}

\begin{figure}
\centering
\begin{minipage}{0.88\linewidth}
\centering
\includegraphics[width=0.9\linewidth]{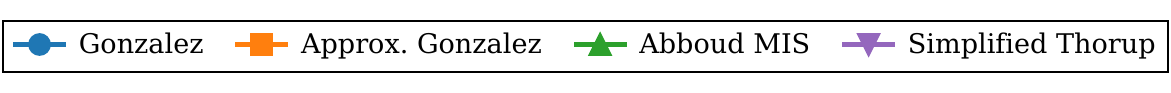}\\[-0.1em]

\begin{subfigure}[t]{0.48\linewidth}
\centering
\includegraphics[width=\linewidth]{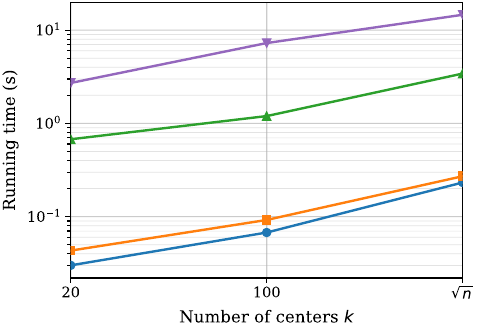}
\caption{libimseti}
\end{subfigure}
\hfill
\begin{subfigure}[t]{0.48\linewidth}
\centering
\includegraphics[width=\linewidth]{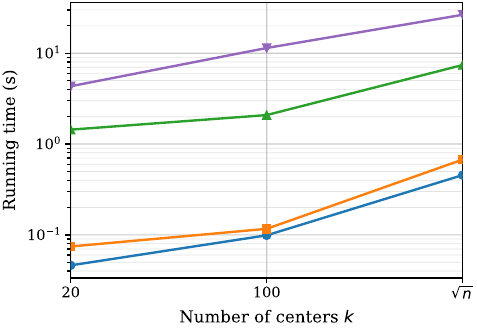}
\caption{MovieLens}
\end{subfigure}

\par\vspace{-0.15em}

\makebox[\linewidth][c]{%
\begin{subfigure}[t]{0.48\linewidth}
\centering
\includegraphics[width=\linewidth]{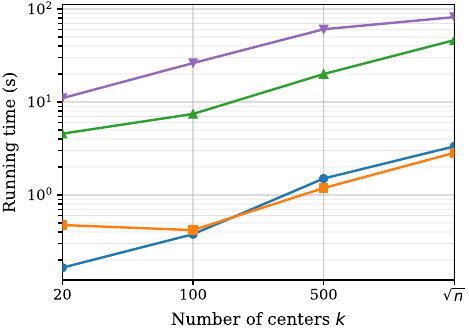}
\caption{Yahoo! Song}
\end{subfigure}
}

\caption{Running times on the three rating graphs. All algorithms use parallel execution with $64$ threads.}
\label{fig:other-full-sweep-runtime}
\end{minipage}
\end{figure}

Lastly, we evaluate our algorithms on three additional real world weighted networks. Unlike the road networks, these graphs are extremely dense and their edge weights originate from transformed ratings, resulting in relatively small edge weights. We observe that in terms of running time, the farthest-point algorithms again come out on top. Although these graphs are very dense, which generally favors the MIS-based approaches, their weighted nature increases the number of candidate radii that must be tested during the feasibility search which seems to offset the benefit from the increased density.

\begin{figure}
\centering
\begin{minipage}{0.88\linewidth}
\centering
\includegraphics[width=0.9\linewidth]{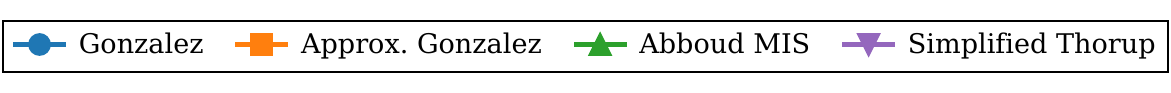}\\[-0.1em]

\begin{subfigure}[t]{0.48\linewidth}
\centering
\includegraphics[width=\linewidth]{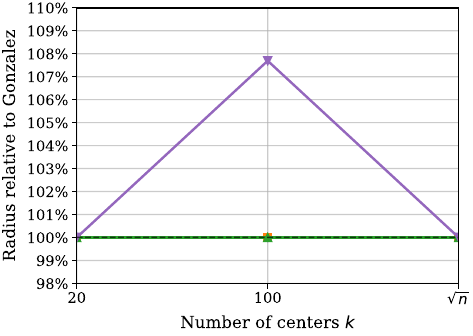}
\caption{libimseti}
\end{subfigure}
\hfill
\begin{subfigure}[t]{0.48\linewidth}
\centering
\includegraphics[width=\linewidth]{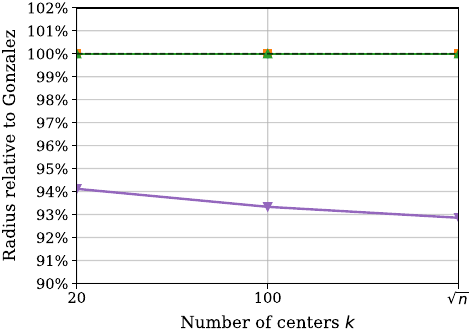}
\caption{MovieLens}
\end{subfigure}

\par\vspace{-0.15em}

\makebox[\linewidth][c]{%
\begin{subfigure}[t]{0.48\linewidth}
\centering
\includegraphics[width=\linewidth]{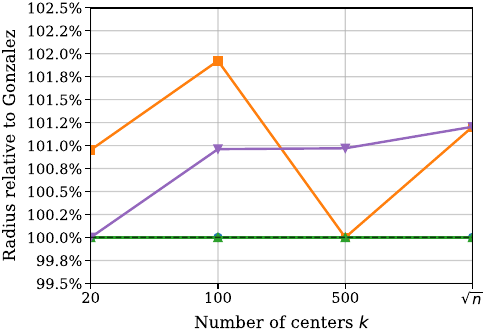}
\caption{Yahoo! Song}
\end{subfigure}
}

\caption{Solution quality on the three additional network instances. The $k$-center solution returned by each algorithm is reported relative to the solution reported by Gonzalez for the same graph and value of $k$.}
\label{fig:other-full-sweep-quality}
\end{minipage}
\end{figure}

The solution qualities of all algorithms are also very similar on these
instances. In absolute terms, the reported radii on libimseti and MovieLens range from approximately $13$ to $17$, while on Yahoo! Song they range from roughly $85$ to $120$ for the different values of $k$. This also yields that the differences shown in Figure~\ref{fig:other-full-sweep-quality} correspond to small absolute differences. For example, the $8\%$ gap between Simplified Thorup and Gonzalez in Figure~\ref{fig:other-full-sweep-quality}, subfigure~(a), for $k=100$ corresponds to an absolute radius difference of only $1$.

\bibliographystyle{siamplain}
\bibliography{references}
\end{document}